\documentclass[10pt,sigconf,letterpaper]{acmart}

\usepackage[english]{babel}
\usepackage{blindtext}
\usepackage{latexsym,subfigure}
\usepackage{pifont}
\usepackage{epsfig,epsf,url}
\usepackage{tabularx}
\usepackage[ruled,vlined]{algorithm2e}
\usepackage{amsmath}
\usepackage{mathtools}
\usepackage{rotating}
\usepackage{wrapfig}
\usepackage{times}
\long\def\comment#1{}
\usepackage{multirow}
\usepackage{algorithmic}
\usepackage{epstopdf}
\usepackage{amsmath,bm}
\usepackage{lscape}
\usepackage{stmaryrd}
\usepackage{wrapfig}
\usepackage{hhline}
\usepackage{textcomp,booktabs}
\usepackage{comment}

\renewcommand\footnotetextcopyrightpermission[1]{} 
\setcopyright{none}

\acmYear{2021}
\copyrightyear{2021}
\setcopyright{acmcopyright}
\acmConference{Submitted for review}
\acmPrice{TBA}
\acmDOI{TBA}
\acmISBN{TBA}

\acmPrice{}

\newcommand{\tele}{Tele-QTP\xspace}
\newcommand{\plain}{TAG-QTP\xspace}
\newcommand{\simple}{TAG\xspace}
\newcommand{\tp}{teleportation\xspace}
\newcommand{\tpqdn}{Tele-QDN\xspace}

\newcommand{\shqdn}{TAG-QDN-S\xspace}
\newcommand{\mhqdn}{TAG-QDN-R\xspace}

\newcommand{\ie}{{\it i.e.}}
\newcommand{\eg}{{\it e.g.}}
\newenvironment{icompact}{
	\begin{list}{$\bullet$}{
			\parsep 1pt plus 1pt
			\partopsep 1pt plus 1pt
			\topsep 1pt plus 2pt minus 1pt
			\itemsep 1.5pt plus 1pt
			\parskip 0pt plus 2pt
			\leftmargin 0.15in}
	}
	{\normalsize\end{list}}

\begin{document}
\title{Quantum Transport Protocols for Distributed Quantum Computing}

\author{Yangming Zhao and Chunming Qiao \\
	Department of Computer Science and Engineering, University at Buffalo}

\renewcommand{\shortauthors}{X.et al.}

\begin{abstract}\label{abs}
Quantum computing holds a great promise and this work proposes to use new quantum data networks (QDNs) to connect multiple small quantum computers to form a cluster. Such a QDN differs from existing quantum key distribution (QKD) networks in that the former must deliver data qubits reliably within itself. Two types of QDNs are studied, one using teleportation (\tpqdn) and the other using tell-and-go (TAG) (TAG-QDN) to exchange quantum data.  Two corresponding quantum transport protocols (QTPs), named Tele-QTP and TAG-QTP, are proposed to address many unique design challenges involved in reliable delivery of data qubits, and constraints imposed by quantum physics laws such as the no-cloning theorem, and limited availability of quantum memory, a precious resource in QDNs.

The proposed Tele-QTP and TAG-QTP are the first transport layer protocols for QDNs, complementing other works on the network protocol stack. Tele-QTP and TAG-QTP have novel mechanisms to support congestion-free and reliable delivery of streams of data qubits by managing the limited quantum memory at end hosts as well as intermediate nodes, under distributed control.
Both analysis and extensive simulations show that the proposed QTPs can achieve a high throughput and fairness. This study also offers new insights into potential tradeoffs involved in using two different types of QDNs.
\end{abstract}

\maketitle

\section{Introduction}\label{sec:intro}
The technologies for building quantum computers are starting to emerge~\cite{jiuzhang,sycamore}. Quantum computing holds the great promise of being able to solve certain types of problems much more efficiently than classical computers~\cite{Grover}, and in fact, some classic NP-hard problems in a polynomial time~\cite{Shor}. However, in a foreseeable future, it is expected that we will be able to build small quantum computers with limited quantum computing power (in term of the number of qubits that can be handled). One way to overcome such a limitation is to network many small quantum computers with a quantum network to support distributed processing~\cite{distributedQC1,distributedQC2}, akin to today's distributed computing systems or cloud systems with classic computers~\cite{vl2,Bcube,DCell}, by enabling one quantum computer to send  quantum state information to another quantum computer. To distinguish such quantum networks specifically designed for supporting distributed quantum computing  from those used for quantum key distribution (QKD), we will refer to the former as quantum data networks or QDNs hereafter. 

In this paper, we investigate new QDN protocols to support dynamic data qubit exchanges between Alice and Bob interconnected with each other via either fiber-based~\cite{singlephoton} or free-space (between a ground station and a satellite)~\cite{Satellite} quantum channels. When Alice and Bob are not directly connected by a fiber, as often is the case in a QDN, an end-to-end quantum connection from Alice (the ingress quantum computer) to Bob (the egress quantum computer) will consist of at least one intermediate all-optical (quantum) switch, trusted relay, or quantum repeater. Hereafter, we will use the term "quantum nodes" to sometimes refer to either quantum computers, trusted relays or quantum repeaters.

In this paper, we propose the first-of-the-kind \emph{transport layer protocols} for QDNs, to be referred to as Quantum Transport Protocols (or QTPs). We note that since the primary functionality of any existing QKD network is to exchange shared secret keys between Alice and Bob, not all qubits transmitted from Alice to Bob in a QKD network need be received properly. This is a major difference between a QKD network and our envisioned QDNs,  and it is this major difference, and the need for an efficient transport layer protocols for QDNs, motivate this work on QTPs. Readers are referred to existing works on protocols addressing issues other than data transport such as network layer protocols~\cite{quantumRouting,infoQuantum,DirectEntangle,conextrouting} and link layer protocols~\cite{qumac}.

To further appreciate the need for the proposed new QTPs instead of using the classical TCP, we first note that 
due to quantum physics laws such as no-cloning theorem~\cite{noclone, nonclone1}, QDNs cannot switch quantum packets as the Internet does for data packets. As a result, a QDN will operate more like a circuit-switched network wherein stream of data qubits are to be transported from one quantum computer  to another.
In addition, a QDN will have to operate in a synchronous (time-slotted) fashion. Finally, a QDN will use the Internet as a separate (overlay) control and management network to send auxiliary messages such as request/response messages to establish an end-to-end quantum connection between Alice and Bob in QDNs.
 

\noindent\textbf{Overview.}
In this paper, we will investigate two types of QDNs, one based on \tp and referred to as \tpqdn, and the other based on direct quantum data exchange using "tell-and-go" (TAG) and referred to as TAG-QDNs. We will describe these two QDNs in more details in the next section. This work focuses on the design and evaluation of corresponding QTPs, namely \tele and \plain, assuming distributed control. In particular, we will focus on two notable design aspects, namely, how to deal with errors in quantum communications
, and how to manage (the limited) quantum memory. While both QTPs also assume distributed flow and congestion control as TCP does, the design choices to be made will distinguish  them from TCP.
We note that for this work on QTPs, it is assumed that data security over established quantum connections in a QDN has already been achieved based on the use of either  QKD or other approaches including post-quantum cryptographers, as well as  security-hardened, tamper-proof physical medium for the quantum connections. 

In the rest of the paper, we design and analyze \tele and \plain,
and compare their performance in terms of fairness, and effective throughput (or goodput). The main contributions are 
\begin{icompact}
	\item although there is existing work on QKD protocols and routing protocols for quantum networks, this is the first work on QTPs with focus on supporting distributed quantum computing;
	\item the proposed \tele and \plain represent two major classes of QTPs and comprehensive analysis and performance evaluation results have been presented.
\end{icompact}

The remainder of the paper is organized as follows.
In Section~\ref{sec:background} we first provide related background including related work. Then, in Section~\ref{sec:design}, we describe the proposed \tele and \plain, and analyze their properties. 
Extensive simulation results showing the efficacy of the proposed QTPs are 
presented in Section~\ref{sec:simu}. We conclude this paper in Section~\ref{sec:conclu}. \emph{This work does not raise any ethical issues.}

\section{Background and Assumptions}\label{sec:background}

As mentioned, we envision a distributed quantum computing system with multiple quantum computers (such as Alice and Bob) interconnected with each other via a QDN and propose and investigate new QTPs. In this section, we briefly discuss the basic aspects of quantum communications and QDNs, including transmission errors (or losses), types of QDNs, time slot duration, and quantum memory requirement, which impact on the designs and performance of the QTPs.

\noindent\textbf{Quantum Transmission Errors.}
Quantum communications is typically based on transmitting photons carrying quantum state information over a quantum channel such as a fiber or free-space optical link. Signal attenuation over a quantum channel, as well as interference, will degrade transmission reliability. In particular, since no amplification or regeneration is allowed in any quantum channel, the quantum transmission error probability may increase exponentially with the fiber length~\cite{DirectEntangle, repeaterless, rateloss, quantumVision}, making a fiber-based quantum channel much less reliable than its counterpart in a classical network. Given that a free-space optical link has a much lower attenuation than a fiber link, it is often advantageous to use two free-space links (an uplink and a downlink) plus a quantum satellite~\cite{Satellite}, instead of a fiber, to connect Alice and Bob when they are far away from each other. In addition, two remote quantum computers will likely use either quantum repeaters or trusted relays as intermediate nodes (and a quantum satellites can fulfil either role).

Hereafter, without loss of generality, we will focus only on fiber links, and omit the discussions on free-space optical links and quantum satellites. In addition, we assume that all messages for control and management of QDNs sent through the auxiliary classical network will be delivered reliably, with a negligible delay.

\noindent\textbf{Teleportation and TAG: Two Quantum Data Exchange Methods} Between Alice and Bob who are connected with a fiber, there are two basic methods for exchanging quantum data information. In the first method, called TAG, Alice generates photons carrying her data qubits, and directly transmits the photons to Bob. Alice also uses the classical network to send information about how she prepared the photons (\eg, which polarization base is used) so Bob can recover the quantum data information from the received photons properly. 
While TAG is simple, and could work efficiently when Alice and Bob are close by and connected with a highly reliable quantum channel, the photons carrying the data qubits could be lost or corrupted due to transmission errors. Since it is infeasible for Alice to keep the exact copy of the data qubits for retransmission due to the no-clone theorem, TAG may not be suitable when the quantum channel between Alice and Bob is not that reliable.

This is when the second method, called teleportation, comes in. To use teleportation, Alice can first generates a pair of entangled photons (also called a Bell pair) using \eg, an entangled photon source (EPS), keeps one of photons to itself and sends the other to Bob. Note that such a photon does not carry any quantum data information so if it is lost on its way to Bob, Alice can regenerate another Bell pair and try it again until Bob receives the other photon and becomes entangled with Alice. Alice then performs a Bell State Measurement (BSM) of her data qubit with her entangled photon, and sends the BSM results (through the classical network) to Bob. Bob can use this BSM result and perform a unitary operation on his entangle photon to receive the quantum data information teleported by Alice. In this way, the data qubit does not need to go through the fiber link, and thus the quantum data exchange between Alice and Bob based on teleportation is much more likely to succeed than simply using TAG.

Note that an EPS may fail to generate a Bell pair and even if it succeeds, one of the Bell pair photons may fail to reach Bob due to transmission error. Accordingly, it is not guaranteed that Alice and Bob can successfully become entangled at will. Instead, it is more like a probabilistic event. If Alice and Bob tries enough times, then they will become entangled (with a high probability).

\noindent\textbf{Alternate Ways to Establish Entanglement Links}
The above describes the basic approach to establishing an entanglement link between Alice and Bob. here, we describe two other alternatives which may be more applicable than the basic approach when Alice and Bob are far apart.
In the first alternative,  an EPS can be placed in the middle, \ie, inbetween Alice and Bob, to reduce the transmission error when it sends a Bell pair of photos, one to Alice and the other to Bob. When both Alice and Bob receive one of the Bell pair photons, they entangled.  The second alternative is more costly in that it requires both $X$ and $Y$ have an EPS. In addition, it requires another device to be placed in the middle, which can perform BSM. Using this alternative, both Alice and Bob generate a Bell pair, and then each keeps one of the Bell pair photons and sends the other to the BSM device in the middle. Once the device in the middle performs a BSM operation successfully, Alice and Bob are entangled. 

Since our focus is on QTPs, hereafter, we will omit detailed discussions on most issues at the physical, link and network layers. In particular, we will focus on the first basic approach to establishing entanglement link and assume that entanglement links can be successfully established when needed. In addition, although EPS can be costly and a scarce resource in a QDN using teleportation, we will focus on the efficient management (reservation and allocation) of quantum memory, as a representative of the scarce resource in a QDN, whether it uses teleportation or TAG.



\begin{figure}
	\centering
	\includegraphics[width=0.72\linewidth]{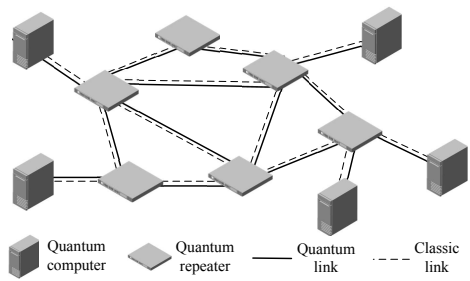}
	\vspace{-0.1in}
	\caption{A Teleportation based Quantum Data Network.}
	\label{fig:sys}
	\vspace{-0.2in}
\end{figure}

\begin{figure}
	\centering
	\includegraphics[width=0.75\linewidth]{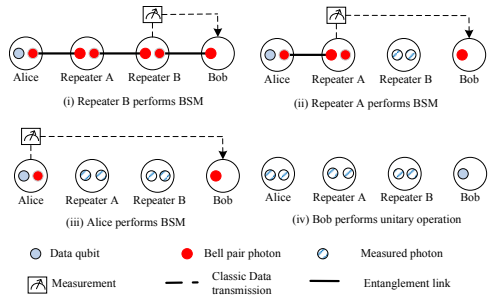}
	\vspace{-0.1in}
	\caption{Teleportation based data qubit transport.}\label{fig:teleportation}
	\vspace{-0.1in}
\end{figure}

\noindent\textbf{Tele-QDNs and Entanglement Connections}
Based on the above discussions, we now describe a \tpqdn with multiple quantum computers and repeaters as shown in Fig.~\ref{fig:sys}, and the concept of teleportation over an entanglement connection is consisting of multiple entanglement links in a \tpqdn.
Assume that an appropriate path from Alice to Bob has been identified (by a quantum routing protocol~\cite{quantumRouting,DirectEntangle, conextrouting}), and every two adjacent quantum nodes 
along the path share an Bell pair of photons (or an entanglement link),  as illustrated by an example shown in Fig.~\ref{fig:teleportation}. 

As can be seen from Fig.~\ref{fig:teleportation}, to establish these entanglement links,  Alice and Bob each needs to store one photon, while each intermediate quantum repeater needs to store two photons belonging to two different Bell pairs. To teleport, Alice, repeaters and Bob all need to perform certain operations and depending on who does first and who does last, there are three primary orders in which these operations are sequenced. 

In this work, given the entanglement links shown, we will perform all the operations needed according to the Repeater-first, Alice-next and Bob-last, or RAB order as follows:
(i). repeater 2 "reads" out the two stored photons from its quantum memory, and performs a BSM operation on them, and sends the BSM results to Bob. This will entangle a photon stored at repeater 1 (shown at its right-hand side) with the photon at Bob, and resulting in the destruction of entanglement links between repeaters 1 and 2, and between repeater 2 and Bob; In addition, repeater 2 has freed up two photon memory. 
(ii) repeater 1 reads out its two stored photons from its quantum memory and performs a BSM operation on them, and sends the BSM results to Bob.  This will entangle the photon at Alice with the photon at Bob, and resulting in the destruction of entanglement links between Alice and repeater 1, and between repeater 1 and Bob; In addition, repeater 1 has freed up two photon memory. 
(iii). After knowing that steps (i) and (ii) have finished successfully, Alice performs a teleportation operation (\ie, a BSM operation of her data qubit with her stored photon), and sends the BSM result to Bob; This changes the quantum state of the photon stored at Bob but also destroys the entanglement between Alice and Bob and frees up a quantum memory at Alice. 
(iv). After Bob receives all three BSM results, he reads out the stored photon from his quantum memory, and performs a series of unitary operations based on the received BSM results. This will enable Bob to receive the data qubit teleported by Alice, and free up a quantum memory at Bob.
As a variation, one may perform step (ii) before or concurrently with step (i), which will not make much difference. In both case, before Alice performs teleportation in step (iii), there is an end-to-end entanglement connection (or circuit) between Alice and Bob.

An alternative to RAB,  referred to as RBAB, is that after each of the first two steps in RAB, Bob will immediately perform a unitary operation based on the BSM results from repeaters 2 and 1, respectively. In this way, after step (iii), Bob only needs to perform a unitary operation based on Alice's BSM result. However, such a variation doesn't offer any meaningful benefit. On the other hand, since the photon at Bob needs to be read-out and stored back multiple times, its fidelity may be compromised.

Last but not least, the third alternative to RAB, referred to as ARB is that step (iii) in RAB is now done first, followed by steps (i) and (ii), which could be in some arbitrary orders, and finally step (iv). Note that in ARB, there is no end-to-end entanglement between Alice and Bob at the time when Alice performs teleportation.
Compared to RAB, a potential advantage of ARB is that by deferring the BMS operations at the repeaters, in case an entanglement link along the path deteriorates due to decoherence for example, teleportation can still succeed if there is a redundant entanglement link over the same fiber link (using a redundant quantum channel) during the same tie slot, thereby achieving a high resilience as discussed in \cite{infoQuantum}. On the other hand, it is also possible that after Alice performs her teleportation, her data qubit may not be able to reach Bob, whereas in RAB, Alice will only perform teleportation after end-to-end entanglement circuit has been established   disadvantage is that the quantum. In addition, with RAB, the repeaters can release their quantum memory earlier for other QTP sessions. We will study this variation in more detail in the future.



\noindent\textbf{Direct Quantum Data Exchange through Tell-and-Go (TAG)}
We now describe the \emph{second} quantum data exchange method called TAG which is applicable to two QDNs of the same type, referred to as \shqdn and \mhqdn. 

Consider a \shqdn with all-optical switches (which can be identical to those used in classical communications). Such a \shqdn can also be illustrated by using Fig.~\ref{fig:sys}, by replacing the quantum repeaters with the all-optical switches.
In a \shqdn, before a data qubit can be exchanged, these all-optical switches must properly configured using an out-of-band signaling to form an all-optical path from Alice to Bob. Then, Alice does the following: (i). prepares a photon carrying the data qubit (by \eg, choosing a polarization basis) for transmission; (ii). uses the classical Internet to inform Bob which basis she has used to prepare the photon (and thus how Bob should measure it later); and (iii). finally sends the prepared photon to Bob along the pre-established all-optical path. When the photon arrives, Bob can receive the data qubit correctly if he measure the received photon properly (\ie, according to how Alice prepared the photon). 

The case in a \mhqdn is similar: assuming that one trusted relay is attached to each all-optical switch, Alice will use TAG to send her data qubit to the first trusted relay in the same way as described before, except that the recipient is now the trusted relay instead of Bob. Then the trusted relay will store the photon carrying the data qubit, and then use TAG to send  the data qubit to the next trusted relay and so on, until Bob receives the data qubit. 

There are two types of trusted relays.  The first is akin to decode-and-forward in conventional wireless networks~\cite{DFNet}, in that their main function is to receive photons carrying data qubits via optical-to-electronic or O/E conversions, and then to regenerate photons (after E/O conversions) carrying the same data qubits for retransmission to the next relay until the photons reach Bob. While this type of relays may be suitable when Alice sends a given bit sequence of 0's and 1's as in the case of QKD, and hence each relay can measure them  and thus regenerate them, these relays are not suitable in QDNs where Alice sends photons, each carrying superposed quantum state information that cannot be measured by a relay. For such photons, their quantum state information must be transferred without going through any operations such as O/E/O conversions that will destroy the quantum state. Accordingly, in this work, we will focus on the second type, which is all-optical (quantum) relays, which use quantum memory to store photons carrying quantum state information, and performs operations such as unitary operations that do not involve any O/E/O conversion of the photons. More specifically, we will consider a special kind of quantum encoding by adapting the concept of quantum secret sharing~\cite{qushare}, with which the presence of these relays will help in increasing the transmission reliability of photons carrying quantum state information, as to be discussed next. 


\noindent\textbf{Discussions on  \shqdn and \mhqdn.}
Although \tp is a lot more exotic concept to networking researchers and as a result, it has received more attention, we consider a \shqdn or  a \mhqdn in this work because the technologies for \simple based quantum data exchanges are relatively more mature than those needed for a \tpqdn. Nevertheless, there are two major challenges associated with \simple: high quantum transmission errors (or losses of photons) and non-applicability of conventional (packet based) retransmission mechanisms for error recovery (both of which are due to the non-cloning theorem).

While a \shqdn is the least expensive and simplest (among the three types of QDNs) to build, control and manage, Alice and Bob will experience a high quantum transmission error probability when they are far away. A trusted relay is able to transmit and receive photons as Alice and Bob do, and can be placed inbetween Alice and Bob. Consequently, quantum transmission reliability between two quantum nodes in a \mhqdn can be improved. However, since the relays are not the true ingress or egress of data qubits, the tradeoffs, besides their high costs associated with the reception and regeneration of quantum state information at the trusted repeater nodes, include a potential degradation in the end-to-end throughput (as data will now have to go through multiple hops, one per relay, from Alice to Bob, and each ··hop" incurs some delays). 

\noindent\textbf{Achieve Reliable Delivery in TAG-QDNs}
While TAG is a familiar concept from classical networks, we note one additional important difference between quantum transmissions and classical transmissions, which will impact on the design of new QTP for \simple-based QDN, named \plain.
In a QDN, when a photon carrying quantum state information is sent, it cannot be copied (or cloned) and later retransmitted for the purpose of achieving reliable transmission~\cite{BB84,nonclone1}. 

To address this challenge which is unique to all QDNs, especially TAG-QDNs, we will investigate a new error management mechanism by adapting an idea from quantum secret sharing~\cite{qushare} in Section~\ref{subsec:plainTCP}. The basic idea is that in a \shqdn or \mhqdn, a data qubit needs to be encoded into multiple data qubits first  and consequently, a \plain needs not only to reserve multiple units of quantum memory at the sender and receiver for each data qubit, but also at the trusted relays in \mhqdn.





\noindent\textbf{Time-slots and Duration.}
Even if one disregards the impact of quantum transmission errors and corresponding mechanisms for error management, different QDNs will have different effects on the throughput of their QTPs due to  other operational differences.
For example, as in many previous works on quantum network protocols, here we also assume that our QDNs operate in a synchronous, time-slotted fashion such that in each time slot, a data qubit transmitted by a quantum node will arrive at another quantum node.
In a \shqdn, this means that the duration of each time slot must be long enough to allow 
(i). a path from each ingress (\eg, Alice) to each egress (\eg, Bob) to be identified; (ii). all the switches along the path to be properly configured to form an end-to-end all-optical connection; (iii). a message containing how Alice prepared its photon to be received by Bob (through the Internet) first, and finally (iv). a photon (carrying a data qubit) to be transmitted by Alice and received by Bob (correctly or not).

In a  \mhqdn, a time slot needs only to be long enough for a trusted relay to receive a data qubit from its upstream node.
Since the (minimum) length of time slot is often proportional to the number of switches along the path, one expects that in a  \mhqdn of a comparable size (wherein one trusted relay is attached to each optical switch in a \shqdn),
the length of each time slot will be shorter than in \shqdn, given that each trusted relay is only one hop away. Nevertheless, it will take multiple time slots for Bob to receive the data qubit from Alice (one per intermediate relay).

The duration of each time slot in a \tpqdn of a comparable size will be the longest among the three. This is because an entanglement link needs to be established between each pair of adjacent nodes along the path from Alice and Bob, which will take much longer than configuring an all-optical switch.  

We note that due to decoherence of entanglement circuit, even the length of a time slot in a \tpqdn should be no more than a second~\cite{qumac,quantumRouting}. In this study, we will assume that in a \tpqdn,  after \tp, one data qubit  will always be received correctly by Bob during each time slot. However,  due to the non-negligible quantum transmission error probability in a \shqdn, multiple time slots may be needed before Bob can receive the data qubit from Alice correctly. Therefore,
it is not clear whether the overall goodput after a period of time in a \tpqdn will be lower than that in a \shqdn (even though 
each time slot in a \tpqdn can be significantly longer than that in a \shqdn).



\noindent\textbf{Quantum Memory Limitation.}
In each of the three QDNs, Alice may concurrently transmit a window of $W$ data qubits to Bob within one time slot, after reserving a sufficient amount of quantum resources at Bob and all intermediate nodes along the way. 

Among the quantum resources, which include photon transmitters and receivers and switch/link capacity, quantum memory~\cite{quantumRouting,qumac} is the most precious. Physically, a unit of quantum memory can be a fiber-optic delay line or an atomic ensemble~\cite{memory}. Current technologies limit the practical size of quantum memory at a quantum node to about one hundred~\cite{memory}. 

Accordingly, we will focus on the impact of the limited quantum memory on the flow and congestion control, reliability and fairness as well as utilization aspects of QTPs. 
As mentioned, the need to use a novel error management mechanism for \simple-based quantum data exchange will distinguish \plain from TCP, whereas the need to reserve a significant amount of quantum memory at all intermediate nodes (quantum repeaters) will distinguish \tele from TCP.

\noindent\textbf{Related Work.}
Except for a few recent work on building a quantum internet~\cite{QuInternet,QuInter,qunature}, all previous works at the protocol level are 
for QKD, whether they use  prepare-and-measure~\cite{BB84,Ekert91,buildQuantum,SECOQC,QKD,DARPA} or \tp~\cite{DARPA,Satellite} to exchange candidate qubits to be used as a shared secret key. Since the goal of any QKD protocol is to receive just a sufficient number of these candidate qubits to build a shared secrete key, instead of exchanging data qubits, not all such candidate qubits sent from Alice need to be received correctly by Bob. As a result, none of the QKD protocols is applicable for a QDN.

All recent works on protocol stack have focused either on the link layer~\cite{qumac} or routing protocols~\cite{quantumRouting,DirectEntangle,conextrouting}. In addition, they have assumed \tp-based quantum data exchanges. To the best of our knowledge, \cite{qustack} was the first work that proposed an protocol-level approach to reliable transmissions based on \simple by leveraging a quantum secret sharing scheme, which however, is insufficient to achieve reliable end-to-end delivery in a QDN based on \simple. Thus, our work presented here is the first concrete work focusing on the transport-layer protocols for both \simple and \tp based quantum data exchange methods.

\section{Protocol Design}\label{sec:design}
In this section, 
we first discuss desired properties of any QTPs in Section~\ref{subsec:property} 
and then describe the proposed \tele and \plain in Section~\ref{subsec:teleTCP} and \ref{subsec:plainTCP}, respectively. 

\subsection{Desired Properties}\label{subsec:property}
One of our goals is to design QTPs that have the following properties:\\
(i). \textit{Reliability:} Due to the no-cloning theory, Alice cannot simply copy a data qubit for retransmission purposes. Accordingly, there should be a protocol-level mechanism to handle photon transmission errors to  ensure reliable delivery of every data qubit. This is particularly important for \plain.\\
(ii). \textit{Congestion-free:}
	Loss may occur due to congestion. Worse, since there is no plausible technique to handle "quantum packets" in a way similar to how data packets are handled in the classical Internet, QDNs will have to work with a stream of qubits. This means that if the qubits from different QTP sessions interleave with each other, and some of them are lost due to congestion, we will not be able to even identify which qubits are lost. Accordingly, we have to design \emph{congestion-free} QTPs.\\
(iii). \textit{High throughput:} One of the major challenges in QDNs for a foreseeable future is a low throughput (in terms of qubits/sec). QTPs need to efficiently manage quantum memory, the most precious resource for both flow and congestion control, in order to  achieve a high throughput. \\
(iv). \textit{Fairness:} Achieving fairness among all traffic flows in a distributed quantum computing system is important. In this work, we aim to maximize Jain's fairness index~\cite{fairness}, a widely used measure of fairness~\cite{Jain1, Jain2}. When the average sending window sizes of all the QTP sessions are $w_1, w_2, \dots, w_N$, the Jain's fairness index is calculated as 
	\begin{equation}\label{equ:fairness}
		J=\frac{(\sum_i w_i)^2}{N\sum_i w_i^2}
	\end{equation}
	The value of the "$J$" index will 1 when all the sessions achieve the same sending window size.\\
(v). \textit{Scalability:} To handle dynamic traffic flows in a distributed quantum computing system, a QTP should be able to quickly and efficiently adapt the sending window sizes and  quantum memory allocations. 


\subsection{Teleportation QTP (\tele)}\label{subsec:teleTCP}
\begin{figure}
	\centering
	\includegraphics[width=0.95\linewidth]{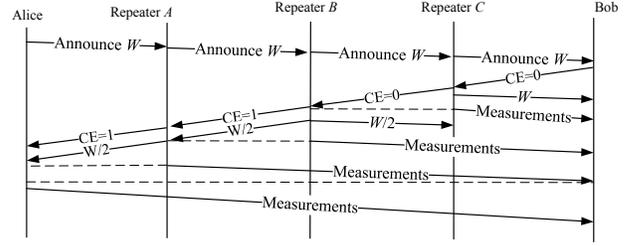}
	\vspace{-0.1in}
	\caption{Time sequence for \tele. Repeater $B$ experiences congestion and sets the CE flag (\ie, CE=1).}\label{fig:timeseq}
	\vspace{-0.2in}
\end{figure}
In this section, we describe the proposed \tele for a \tpqdn. As mentioned earlier, in  a \tpqdn, Alice needs to first establish $W$ entanglement circuits  in a time slot with Bob, if Alice wants to teleport $W$ qubits (because only one qubit can be teleported per circuit, after which the circuit will be destroyed). This means Alice and Bob each needs to use at least $W$ qubits of quantum  memory, and at least $2W$ qubits of quantum  memory at each of the intermediate repeaters. Unlike in a classical network, where congestion may occur during the actual data transmissions of TCP segments for example, here, the main challenge becomes how to establish these circuits given the limited quantum memory at each repeater, since once these circuits are established, teleportation of actual data qubits will not experience any congestion (and thus has a negligible loss as mentioned earlier).

\emph{Distributed vs Centralized Reservation}: Given that there may be a large number of  sources (Alice), each of which wants to establish a large number of circuits in a time slot, and it is desirable to establish these circuits as quickly as possible (due to limited duration of a time slot), we propose a distributed reservation (or signaling) protocol to circuit establishment.

In a typical implementation of a distributed reservation approach, each source sends a request (in the classical network) for  $W$ circuits along a path to the corresponding destination (Bob), where the path is determined by a quantum routing protocol. Such a distributed protocol typically has two phases: a forward phase and a backward phase. In the forward phase, each request is processed hop-by-hop by local controllers co-located with the quantum repeaters to see if $W$ entanglement links can be established and then forwarded to the next hop (controller) for further processing until it reaches the destination. At this time, the forward phase ends and the backward phase will start. 

One possible way to make reservation for the requested $W$ circuits is for each local controller to reserve as many quantum memory units (or equivalently to establish as many entanglement links) as possible during this forward phase. Such an approach would be  similar to the so-called source-initiated reservation (SIR) protocol, proposed in \cite{sir-dir} for  establishing an all-optical lighthpath (or wavelength path) in a Wavelength Division Multiplexed (WDM) network. 
An alternative is that, during the forward phase, each local controller only collects the information on the (maximum) number of quantum memory units (or equivalently entanglement links) that are needed in a hop-by-hop fashion. Then in the backward phase, starting from the destination, quantum memory at each node is allocated (and entanglement links are established). This alternative is similar to how  the destination-initiated reservation (DIR) protocol~\cite{sir-dir} works.  In this paper, we will focus on an approach which closely resembles DIR. Readers are referred to ~\cite{sir-dir} for discussions on the pros and cons of the SIR and DIR protocols.

Note that, using such a distributed reservation approach,  the amount of time it takes to establish circuits for all SD pairs is proportional to the maximum round-trip time, over all SD pairs, between a source and its destination (which includes both the total propagation time and the sum of the hop-by-hop request/response processing time), but is largely independent on the number of SD pairs. On the other hand, if one uses a centralized reservation approach whereby a central controller processes all the requests to derive the reservation decisions, then the amount of time needed to finish processing all the requests could be longer, since this will not only include the propagation time between the sources and the central controller, and between the central controller and the repeaters, but also become proportional to the total number of requests. Nevertheless, with the advance in Software Defined Networks (SDN), centralized control, which could potentially achieve a more efficient resource utilization and a higher network-wide throughput, is also very promising,  and thus will be left as a topic for our future study.

\emph{A TCP Friendly Heuristic for Distributed Reservation}: In order to improve as much as possible both the throughput (in terms of the number of circuits that can be established per source) and fairness among difference sources under distributed control, we propose to adopt a TCP friendly heuristic reservation approach under the general framework of DIR. The first design goal of this distributed reservation protocol is that during each time slot, the protocol should prevent some SD pairs from monopolizing the resources. In other words, the protocol should guarantee that any source who wants to establish a circuit can do so, as long as it is feasible, without being starved due to aggressive behaviors by other sources. The second design goal is to deal with unknown but dynamically changing demands from various sources effectively so as to guarantee that over time, (or asymptotically),  any source can achieve its maximal throughput subject to the network-wide fairness requirement.  The third design goal is to enable each local controller to be able to decide on how much quantum memory to be allocated to a given SD pair, and take appropriate actions (by establishing entanglement links and performing BSM) during the backward phase, without having to go back and forth,  in order to speed up the circuit establishment process.

We observe that although TCP was designed for a packet-switched network and hence can't be directly applied here,  it still has a tried-and-true congestion mechanism that provides many desirable features. Accordingly, we can adapt a TCP-based mechanism to meet the above design goals. More specifically, in this paper, we use an Additive Increase and Multilicative Decrease (AIMD) based approach to determine how many entanglement links are to be committed to each source at each hop (during the backward process). As in a AIMD based TCP congestion control mechanism, the basic idea is to allow each source to request say $W$ circuits in one time slot and $W+1$ circuits in the next if possible. If/when there isn't a sufficient amount of quantum memory at a given hop to satisfy all the requests from different sources,  the protocol will cut the number circuits granted to a given source by half.  In order to facilitate the presentation, hereafter, we use the terms common in TCP congestion control such as "sending window size" and "congestion encountered", respectively to refer to the number of circuits to be established, and to indicate whether the total number of requested circuits through a given repeater has exceeded the number of quantum memory it has or not. Note that other heuristics can also be used, and we leave them for future studies.


\subsubsection{Time Sequence in \tele}

Fig.~\ref{fig:timeseq} illustrates the two-phase (forward and backward) process in \tele  to establish $W$ circuits (using the auxiliary Internet for signaling messages).

During the first step (\ie, the forward phase), Alice announces the desired sending window size, $W$, (or equivalently a request to reserve $W$ circuits) along a pre-selected path (by some routing protocols, such as~\cite{quantumRouting, infoQuantum}). After receiving this announcement (or request), every repeater along the path takes note of the requested window size, and waits for similar announcements from other QTP sessions to arrive. Note that every 
source who wants to establish entanglement circuits in this time slot must either send out their announcements at the start of the time slot (or at the beginning of the next time slot). Assuming that it takes at most $L$ units of time for an announcement from Alice to reach Bob, then no other destination can start Step Two earlier (even though it has already received all the announcements from other ingress nodes).

In Step Two (the backward phase), Bob initializes a response message and sends it back to Alice (along the same path but in reverse direction). Similar to DCTCP~\cite{dctcp} and DCQCN~\cite{DCQCN}, the response message contains a Congestion Experienced (CE) flag to indicate if there is congestion along the selected path. The sending window size will be cut to half (\ie, $\lfloor W/2 \rfloor$) if and only if the CE flag is set by any node.

More specifically, when a node receives (or Bob initializes) a response message, it will first check if the CE flag has been set already or not. If so, it forwards the response message to its upstream node, and tries to set up $\lfloor W/2 \rfloor$ entanglement links with its upstream node (towards Alice). If the CE flag is not yet set, the node will determine if it can accommodate all circuits requested by all sources that will go through this node. 
In other words, for each and every QTP session requesting $W$ circuits going through this node,  
it will determine if it could allocate a total of $2W$ units of quantum memory for the QTP session. If it could, it will establish $W$ entanglement links with its downstream quantum node (towards Bob) and forwards the response message without setting the CE flag.
Otherwise, this QTP session's sending window size will be cut to half (as to be explained, \tele has built-in mechanisms to ensure that all nodes along the path can handle at least $\lfloor W/2 \rfloor$ entanglement links per hop for this QTP session during this time slot. Accordingly, none of the other upstream repeater nodes will need to cut the sending window size any further. More specifically, this repeater node will do the following:  (i). set the CE flag, and send the response to an upstream node; (ii). allocate $\lfloor W/2 \rfloor$ units of quantum memory for each direction (upstream and downstream); and (iii). establish $\lfloor W/2 \rfloor$ entanglement links with both its upstream and downstream nodes.
In any case, once an entanglement link with its upstream node has been established, the repeater will perform a BSM operation and send the BSM result to Bob. Such an operation can take place while  other upstream repeaters process the response message.
Meanwhile, Bob will store the BSM result, along with the BSM results from all other repeater along the path, and use them when Alice performs a teleportation and sends her BSM results to Bob.

For Alice and Bob, there are two exceptions: (i). As the source and destination of an entanglement path, Alice and Bob do not need to set up entanglement links with their upstream and downstream node, respectively; (ii). Bob only needs to determine if it could allocate a total of $W$ units of quantum memory for this QTP session.

\subsubsection{Flow and Congestion Control}

In this subsection, we describe how Alice controls the sending window, and how each node assigns quantum memory to each QTP session. 

\noindent\textbf{Sending Window Control (SWC):}
\begin{algorithm}[t]
	\caption{Sending Window Control (SWC)}
	\label{alg:sender}
	\begin{algorithmic}[1]
		\setlength\rightskip{1em}
		\REQUIRE Number of qubits to be sent $Q$
		\STATE Initialize sending window $W\gets 1$, state $S \gets SS$ \label{lin:init}
		\WHILE{$Q>0$}
			\STATE Announce sending window size $W$ \label{lin:reserve}
			\IF{$CE = 1$}
				\STATE $W \gets \lfloor W/2 \rfloor$, $S\gets CA$ \label{lin:ce}
			\ENDIF
			\STATE Try to establish  $\min\{W,Q\}$ entanglement circuits, $Q\gets Q-\min\{W,Q\}$ \label{lin:teleport}
			\IF{$S=SS$}\label{lin:beginIncre}
				\STATE $W\gets 2\times W$
			\ELSE
				\STATE $W\gets W+1$
			\ENDIF\label{lin:endIncre}
		\ENDWHILE
	\end{algorithmic}
\end{algorithm}
To ensure congestion free, every QTP session will use the proposed Algorithm~SWC. Similar to TCP, \tele has two states for each QTP session, \ie, slow start (SS) and congestion avoidance (CA). When a new QTP session starts, Alice is in the state SS, and the initial sending window size $W$ is set to 1 (Line~\ref{lin:init}). In the next time slot, after $W$ qubits are teleported, the sending window size will double if the session is in the state SS, or  increase by 1 if the session is in the state CA (Lines~\ref{lin:beginIncre}--\ref{lin:endIncre}). Any time Alice receives a response with the CE flag set, the sending window size will reduced by half and  the state will be set to CA (Line~\ref{lin:ce}).  In this way, the QDN will be congestion-free (see detail in Theorem~\ref{the:congestion}).

\noindent\textbf{Quantum Memory Assignment (QMA):}
\begin{algorithm}[t]
	\caption{Quantum Memory Assignment (QMA)}
	\label{alg:repeater}
	\begin{algorithmic}[1]
		\setlength\rightskip{1em}
		\REQUIRE The quantum memory requirement from session $n$, $W_n$ (Suppose there are $N$ sessions and $W_1 \geq W_2 \geq \dots \geq W_N$); the total amount of quantum memory it hosts, $M$; the amount of quantum memory required to teleport one qubit, $a$
		\STATE $A \gets \sum_n aW_n$, $n\gets 1$
		\IF{$A \leq M$}
			\STATE Reserve $aW_n$ units of quantum memory to session $n$ for all sessions \label{lin:noce}
		\ELSE
			\WHILE{$A>M$}\label{lin:beginhalf}
				\STATE Mark $CE \gets 1$ for session $n$ \label{lin:mark}
				\STATE Reserve $a\lfloor W_n/2 \rfloor$ units of memory for session $n$
				\STATE $A\gets A-a\lfloor W_n/2 \rfloor$, $n\gets n+1$
			\ENDWHILE\label{lin:endhalf}
			\WHILE{$n\leq N$}
				\STATE Reserve $aW_n$ units of memory to session $n$
			\ENDWHILE
		\ENDIF
	\end{algorithmic}
\end{algorithm}
Algorithm~QMA, which should be deployed on every node (\ie, both repeaters and quantum computers), shows how each node in \tele to manage the precious quantum memory. The main idea is that each node will first try to satisfy all the quantum memory requests received in the first step of the above described 
teleportation process (Line~\ref{lin:noce}). If there is not enough quantum memory,
the node fulfills only half of the quantum memory requirement for selected sessions (Lines~\ref{lin:beginhalf}--\ref{lin:endhalf}), starting with the one which requires the largest sending window (Line~\ref{lin:mark}), until either it can satisfy the requested window sizes from all remaining sessions, or have cut all the requested window size by half.

It should be noted that while Bob only needs to reserve $W$ units of quantum memory to support a QTP session of window size $W$, all other nodes (\ie, Alice and all repeaters along the path) have to reserve $2W$ units of quantum memory. To ensure fairness among multiple QTP sessions, in \tele, we divide the quantum memory for data exchanges at each quantum computer (which could play the role of Alice for some sessions, and Bob for other sessions) into two parts: one-third for receiving data qubits  and two-thirds for sending data qubits. Correspondingly, there would be two processes of Algorithm~QMA on each quantum computer: one manages the memory for receiving data qubits ($a=1$), while the other manages the memory for sending data qubits ($a=2$).

\noindent\textbf{Variations of \tele:}
While there are many possible variations and potential improvements, we present two alternatives to \tele which we will compare via simulations in a later section:\\
(i). \textbf{Explicit Window (EW):} This alternative aims to achieve fairness at each and every node and works as follows. As in \tele, in each time slot, Alice will send a request to Bob for a QTP session along a pre-determined path. However, Alice does not (need to) specify the sending window size at all.  After Bob and all intermediate nodes have received all the requests, each node capable of supporting a sending window of size $C$ (\eg, either Bob having $C$ units, or an repeater having $2C$ units) will assign a window size of $C/N$ to each of the $N$ sessions requested of the node. Alice then uses the smallest window size among all the nodes along the path for the current time slot. \\
(ii). \textbf{Fair Resource Allocation (FRA):} This alternative aims to become a hybrid of \tele and EW and works as follows. As in \tele, in each time slot, Alice will request a window size and send the request to Bob for a QTP session along a pre-determined path. Different from \tele, Bob and each repeater simply checks if the requested window size is larger than $C/N$ or not, and if so, it sets the CE flag (which will force Alice to cut its window by half).





\subsection{Tell-and-Go QTP (\plain)}\label{subsec:plainTCP}
In this subsection, we first describe \plain for a \shqdn, wherein only Alice and Bob (but no intermediate nodes) have to reserve quantum memory. However, given lossy quantum channels (and the no-cloning theorem which invalidates the conventional restransmission-based approaches), \plain has to address two coupled problems: reliable end-to-end delivery, and the corresponding quantum memory allocation (at both Alice and Bob).

\noindent\textbf{Reliable End-to-End Delivery:}
To provide reliable end-to-end delivery of data qubits, \plain borrows the idea from quantum secret sharing~\cite{qushare}, whose main idea can be summarized as follows. For each data qubit, we encode it with $N$ qubits (each one of these qubits will be called a \emph{sharing} to distinguish it from the original data qubit),  such that we can recover the original data qubit with any $K$ of these $N$ sharings. We refer to such a method as $(K,N)$-threshold sharing. Although quantum secret sharing was originally proposed for an application where there are multiple ($N$) parties, to ensure that the original data qubit can be recovered if and only if $K$ or more parties agree to collaborate on the recovery operation, here, we adapt it to ensure reliably delivery of a data qubit from Alice to Bob, in a way similar to, but different from the classical error correction code (FEC)~\cite{fec}, as to be explained next.


It has been shown that for a data qubit with larger than or equal to 3 dimensions (in terms of quantum states), one can use a (2,3)-threshold sharing scheme, while for a 2-dimensional qubit, one needs to use a (3,4)-threshold sharing scheme~\cite{qushare}. Below, we discuss the case with a (2,3)-threshold sharing scheme but note that our method can be easily extended to the case with a (3,4)-threshold sharing scheme. 

Let the sharings for a data qubit, $A$, be denoted by  $A_1^{0}A_2^{0}A_3^{0}$. Alice will first send $A_1^0$ (in the first time slot). 
If $A_1^0$  fails to arrive at Bob, Alice can recover the original data qubit $A$ (without violating the no-cloning theory) since she still stores $A_2^0A_3^0$. Afterwards, Alice repeats the process all over again in the next time slot.

If $A_1^{0}$ is successfully delivered,  Bob will store it (but can't measure it for now) and Alice will send $A_2^{0}$ in the next time slot. If $A_2^{0}$ is also successfully delivered, Bob can then 
re-construct the original $A$ with $A_1^{0}A_2^{0}$, and this completes the reliable delivery of $A$. Otherwise, Alice can encode the remaining $A_3^{0}$ with a (2,3)-threshold sharing scheme and try to deliver the three sharings corresponding to $A_3^{0}$, to be denoted by $A_1^{1}A_2^{1}A_3^{1}$ using a similar process. If (and only if) $A_3^{0}$ eventually gets reconstructed by Bob (in a recursive process), Bob can then reconstruct the original data qubit $A$ with $A_1^{0}A_3^0$. To facilitate the following discussion, we will refer to 
sharings of the form $A_1^{k}$, $A_2^{k}$ and $A_3^k$ for any $k$ as the first, second and third sharing in the $k$th round, respectively. Appendix~\ref{app:sharing} illustrates the above process using a state machine.


The above discussion implies that Alice has to reserve 3 units of quantum memory to store 3 sharings for each data qubit. 
In addition,  Bob needs at least 2 units of quantum memory to receive data qubit $A$ (this  best case occurs when $A_2^{0}$ arrives successfully). Note that, however, whenever sharing $A_2^{k}$ ($k\geq 0$) is lost, Bob will need 
$k+3$ units to store sharings $A_1^{i}$ for all $0 \leq i \leq k$ plus $A_1^{k+1}$ and $A_2^{k+1}$. All these units of quantum memory can be released if and only if sharing $A_2^{k+1}$ arrives at Bob.

When $k$ becomes large, Bob may run out of the $W$ units of quantum memory reserved for this session, and consequently,  not only Alice will be unable to send $A_2^{k+1}$, thus preventing Bob from releasing any of its quantum memory, but Bob cannot send to other nodes as well, potentially resulting in a deadlock.
This is why it is critical for \plain to manage quantum memory for Alice and Bob.




\noindent\textbf{Quantum Memory Allocation:}
In \plain, Alice can still leverage Algorithms~SWC and QMA to control the sending window and assign quantum memory to each QTP session, respectively. However, several changes are needed. First, for reasons to become clear, the initial window size for each session will be set to $W=2$ instead of 1 (even though Alice may only have one data qubit to send). Secondly, while in the CA phase, $W$ will increase by one after each time slot, regardless of how many data qubits (or sharings) have been successfully delivered in the current time slot.

In addition, \plain needs to make the following major modifications to support reliable delivery of data qubits encoded using sharings: \\
(i). Both Alice and Bob have to maintain the state machine for each and every data qubit and different data qubits will be in different states, characterized by the most recent (or largest) value of $k$ associated  with their corresponding first sharings that have been successfully transmitted by Alice and stored at Bob. For example, after the previous time slot, Alice has successfully transmitted $i$ first sharings corresponding to data qubit $X$, and thus will need to transmit the second sharing of the $i$-th round next; meanwhile, Alice has failed to transmit the second sharings of the $j$-th round corresponding to data qubit $Y$, and will need to transmit the first sharing of the $(j+1)$-th round next.
During this time, we note that  Bob stores all $i$ first sharings corresponding to data qubit $X$, and all $j$ first sharings corresponding to data qubit $Y$;\\
(ii). Continuing from the above example, if the second sharing of the $i$-th round for $X$ arrives at Bob, Bob can release all $i+1$ units of the quantum memory allocated to the sharings corresponding to $X$. On the other hand, even if Alice transmits the first sharing of the $(j+1)$-th round successfully to Bob, Bob won't be able to release any quantum memory allocated to $Y$. Instead, Bob will have to store this along with the existing $(j+1)$ first siblings for $Y$.

Based on the above observation, in order to enable Bob to release the maximal number of quantum memory units as soon as possible, Alice should first select the data qubit with the largest $k$ (say it's A), for which the corresponding first sharing  of the $k$-th round has been stored at Bob, and send the second sharing of the $k$-th round (\ie, $A_2^k$), in hopes of enabling Bob to release all stored first sharings for $A$; Alice will try to send as many of these second sharings as possible,
subject to the availability of the quantum memory available at Bob.\\
%
(iii) Assuming that at the beginning of a time slot, Bob has stored a total of $s$ first sharings for different data qubits from the current sending window, and thus can receive $(W-s)$ additional sharings in this slot. Further assume that corresponding to these $s$ first sharings at Bob, Alice have only  $a_2 \leq s$ second sharings to send ("less than" is possible as the rest of $(s- a_2)$ second sharings have already been lost in the previous time slot, as assumed for data qubit $Y$ described above). In such a case, Alice will send out
as many first sharings as possible starting with those sharings having the largest $k$ and subjecting to the amount of quantum memory available at Bob.


Note that, the sending window of every session will be periodically cut by half. To prevent the number of sharings stored by Bob is larger than the window size, Alice will try to limit the number of the first sharings Bob has to store at the end of each time slot to $\lfloor W/2 \rfloor$. To this end, Alice puts an upper bound on the number of the first sharings Alice can send, which is $a_1=\max\{0, \lfloor \frac{W}{2}\rfloor +  a_2 - s\}$.  Intuitively,this implies that at the start of the session, Alice will send only $W/2$ first sibling (this is why we will set the initial window size to $W=2$). Once the window size is less than the first sharings stored at Bob, corresponding QTP session will be close to avoid congestion.\\
(iv). Alice reserves quantum memory to each session based on the number of data qubits that has be encoded as sharings rather than the window size $W$. In general, we have the following three desired constraints: (1). $a_2 \leq s$; (2). $s + a_1 + a_2 \leq W$; and (3). $a_1 + (s - a_2) \leq W/2$, we can use linear programming to derive equation (4).  $a_1 + a_2 \leq 3W/4 $. This equation means that during each time slot, Alice can send different sharings corresponding to at most $3W/4$ data qubits, each of which requires Alice to reserve 3 units of quantum memory. Consequently, this means for a sending window of size $W$, Alice needs to reserve a total of $9W/4$ units of quantum memory at most, which is fewer than $3W$. 

Similar to the case \tele, since a quantum computer may serve the role of Alice (or ingress) for some sessions and Bob (or egress) for other sessions in a distributed quantum computing system,  we need to partition the total amount of quantum memory available at each quantum computer 
into two parts:
one part containing $9/13$-th of the total  for sending sharings, and other containing the remaining $4/13$-th of the total for receiving sharings. For \plain, this is important for avoiding deadlock while supporting reliable delivery.
%
%

\noindent\textbf{Extension to \mhqdn:}
As noted earlier, \shqdn, though simple, may not be scalable when quantum computers are distributed in a large  geographic area. To support reliable delivery when Alice and Bob are far from each other, \mhqdn can be leveraged if \simple is preferred over \tele, as the method  for quantum data exchanges. Nevertheless, the above described \plain can be applied to not only Alice and Bob, but also all the trusted relays in a \mhqdn.
Since each trusted relay has to receive photons from an upstream node, decode the quantum information they carry, generate new data qubits, and send them to the downstream node, it will have to divide their quantum memory into two parts as described in (iv) above.

\subsection{Analysis}\label{subsec:analysis}

Based on the algorithm design, we know \tele and \plain are scalable and reliable. In this section, we prove that both of them can achieve congestion-free, fairness, and high quantum memory utilization. We will focus on \tele, although similar analysis apply to \plain too since the latter is an extension of \tele. 

In the analysis, we consider $N (N\geq 2)$ infinitely long-lived \tele sessions which share (or compete for) a limited amount of quantum memory at a bottleneck quantum node whose maximum capacity is to support the \tp of $C (C \geq N)$ data qubits concurrently (in the same time slot). More specifically, if Bob is the bottleneck node, then this means Bob has  $C$ units of quantum memory to receive data qubits, whereas if a quantum repeater is the bottleneck node, then the repeater has $2C$ units of quantum memory.

Due to the space limit, the solid proofs of all the theorems below are omitted but can be found in Appendices~\ref{app:congestion}--\ref{app:conserving}.
\begin{theorem}\label{the:congestion}
	(Congestion-free) During any time slot, there is always enough quantum memory to support quantum data exchanges over all QTP sessions, as long as every ingress quantum computer (Alice) sends data qubits following its sending window size for the time slot.
\end{theorem}
\begin{theorem}\label{the:fairness}
	(Fairness) In a stead state, the window size of every QTP session will be sawtoothly varying within the range of $[W^*/2, W^*]$, with the average window size being $3W^*/4$, where $W^*$ is the maximum window size determined by $N$ and $C$.
\end{theorem}
We note that the above Theorem shows that all the QTP sessions share same amount of quantum memory on average, \ie, the Jain's index will be $J=1$.
\begin{theorem}\label{the:converving}
	(High quantum memory utilization) During any time slot, at most $200/3N$ percent of the quantum memory will be sitting idle.
\end{theorem}
We note that when the number of QTP sessions approaches infinity, we have $\lim_{N\to \infty}\frac{200}{3N} = 0$. In other words, the QTPs proposed in this work can fully utilize the quantum memory at the network bottlenecks. 

To validate the correctness of our analysis, we also conducted simulations in Appendix~\ref{app:validation}.

\section{Performance Evaluation}\label{sec:simu}
In this section, we first evaluate the performance of the proposed \tele in a \tpqdn and \plain in a \mhqdn through extensive simulations. 
We then compare the performance  between \tele in a \tpqdn and \plain in a \shqdn in order to seek additional insights on some of the design tradeoffs involved in these two types of networks.

\subsection{Performance in Wide-area Networks}\label{subsec:macro}
In this section, we conduct extensive simulations to show the performance of proposed QTPs. 
Since \shqdn may not be suitable for a wide-area network,  we will focus on \tpqdn and \mhqdn, both of which use some kind of intermediate quantum nodes, and evaluate the performance of \tele, its two variations (EW and FRA), and \plain.  For the time being, we assume that due to the use of trusted relays in a \mhqdn, every sharing can be successfully delivered to the next node with probability 1.

To generate a network topology to simulate, we randomly place a given number of repeaters (or trusted relays) into a square area.
Edges among repeaters are determined following the Waxman model~\cite{Waxman} by assuming an  average node degree of 4. Each node in the network has a total of 1000 units of quantum memory. In a \tpqdn, each quantum computer reserves 2/3 of the quantum memory (\ie, 667 units) for sending data qubits, while the remaining 1/3 are reserved for receiving data qubits. In a \mhqdn, the split is 9/13 vs. 4/13.

Without loss of generality,  we implement a rudimentary load balance routing algorithm~\cite{loadbalance} to calculate the path for each QTP session. 
To investigate the performance of different QTPs, we record the sending window size of every session during each time slot, utilization of quantum memory 
at every node, and overall network throughput.

In a \mhqdn, since the sending window size of a given "end-to-end" session (Alice to Bob) may be different on every hop,
we only consider the \emph{effective window size}, which is defined as the minimum window size along its path. Typically, a larger sending window size would lead to a higher quantum memory utilization as well as overall throughput, which for a given session, is calculated based on the total number of data qubits received by Bob over a period of time.



\noindent\textbf{Effect of network size:}
At first, we investigate how \tele and \plain perform in networks of different sizes. To this end, we vary the number of repeaters (or trusted relays) in a network from 40 to 70 (in increment of 10) while injecting 100 QTP sessions into these networks of different sizes. 
Each simulation runs for 200 time slots.

Figs.~\ref{fig:sizeWin}--\ref{fig:throuNode} show, respectively, simulation results on how the average sending window size, quantum memory utilization, and network throughput will change with the network size. 

\begin{figure*}
	\subfigure[\tpqdn with \tele.]{\label{subfig:nodeWinTele}
		\includegraphics[width=0.23\linewidth]{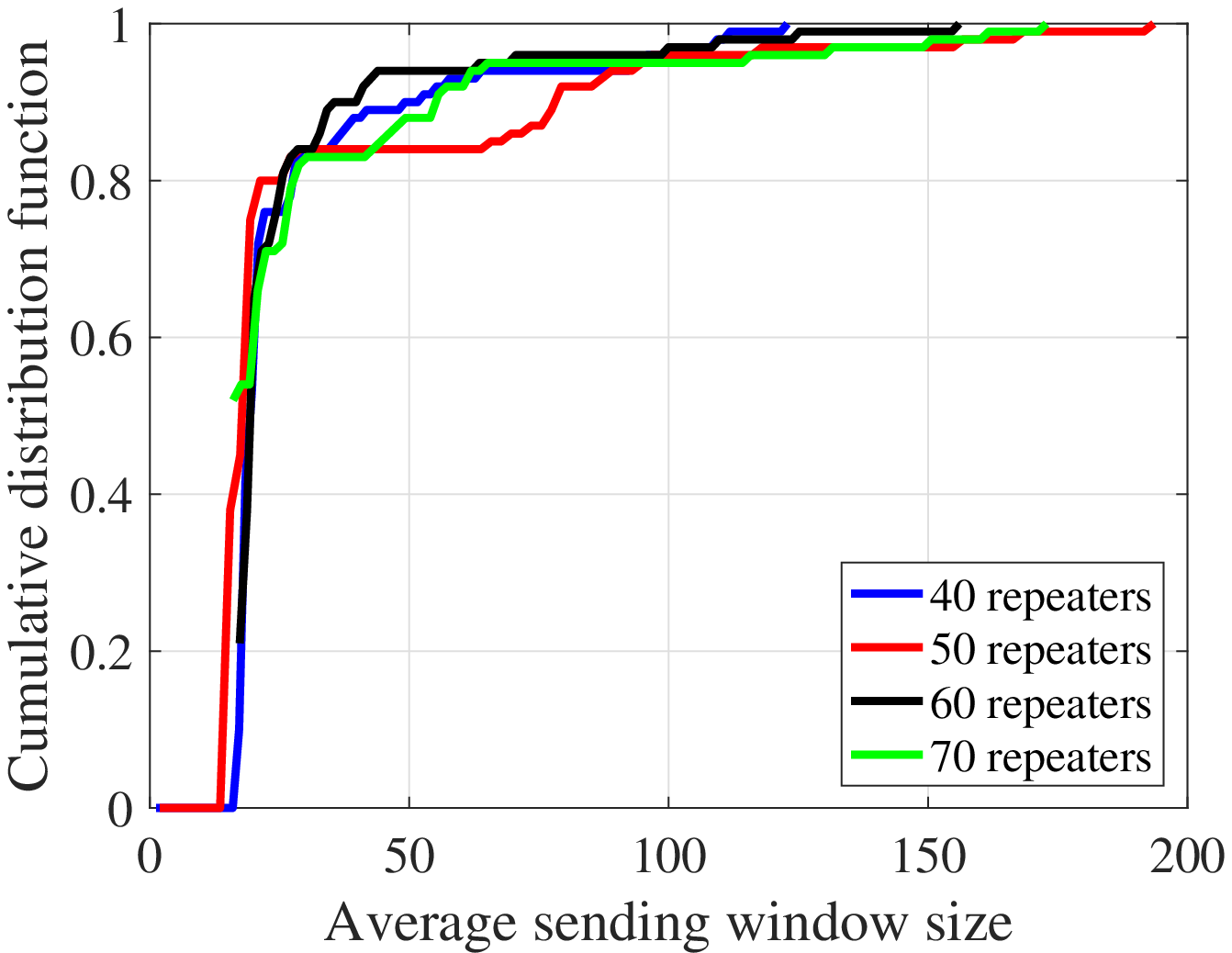}
	}
	\subfigure[\tpqdn with EW.]{\label{subfig:nodeWinExplict}
		\includegraphics[width=0.23\linewidth]{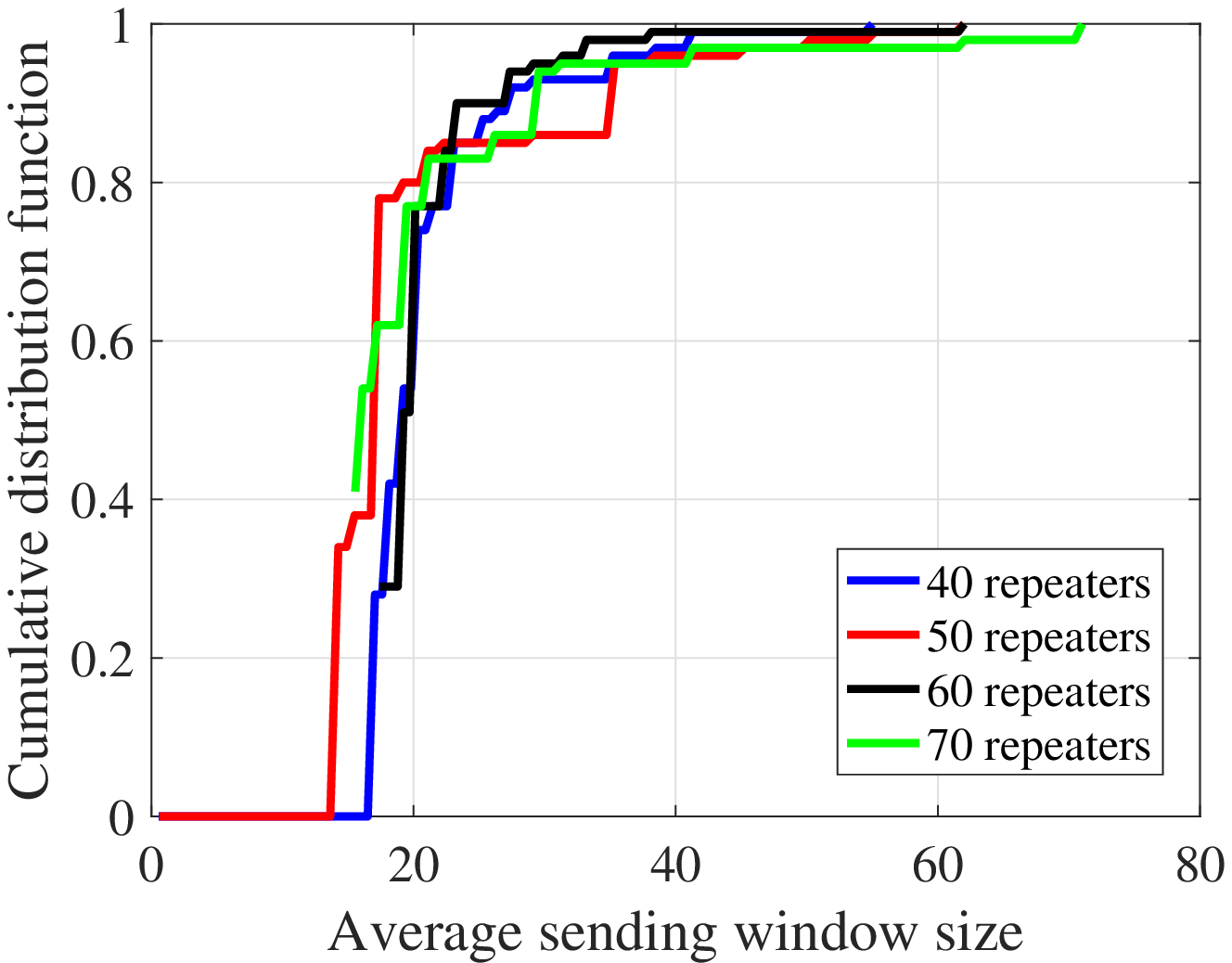}
	}
	\subfigure[\tpqdn with FRA.]{\label{subfig:nodeWinFair}
		\includegraphics[width=0.23\linewidth]{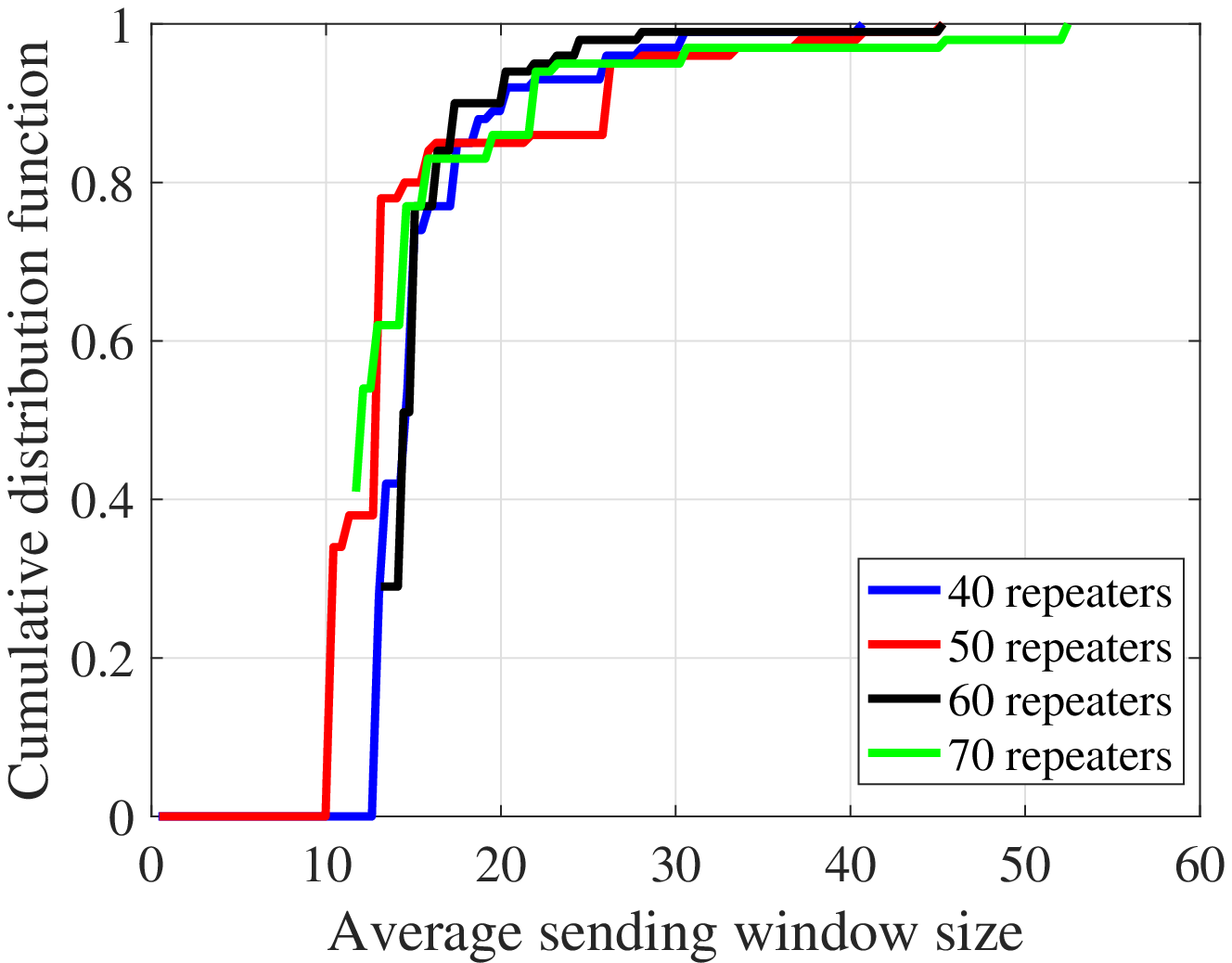}
	}
	\subfigure[\mhqdn with \plain.]{\label{subfig:nodeWinTR}
		\includegraphics[width=0.23\linewidth]{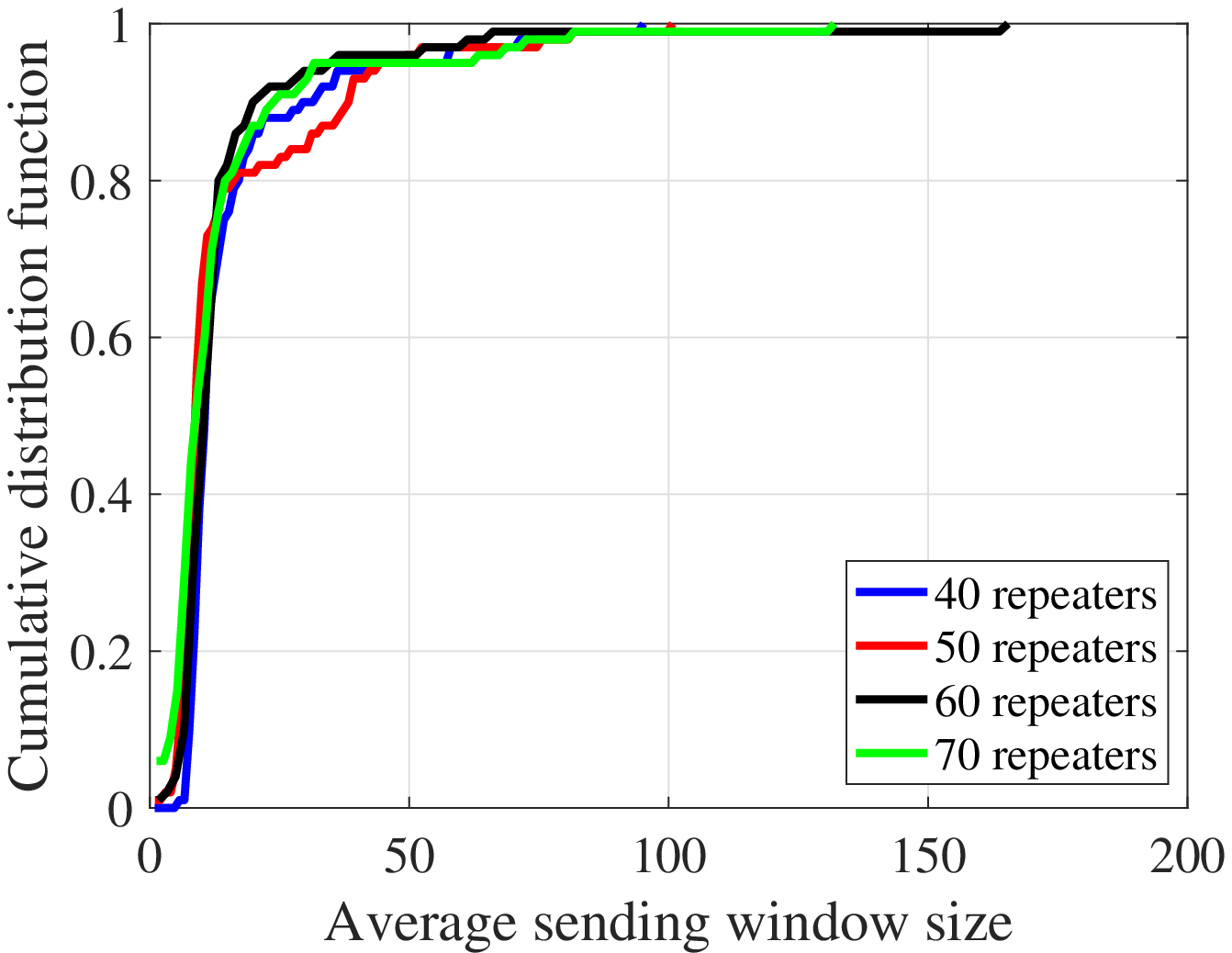}
	}
	\vspace{-0.1in}
	\caption{How network size impacts average sending window sizes.}\label{fig:sizeWin}
	\vspace{-0.2in}
\end{figure*}

As can be seen from Fig.~\ref{fig:sizeWin}, the average sending window size does not vary too much with the network size. On the one hand, a smaller network means a fewer hops between Alice and Bob, potentially leading to a larger sending window. On the other hand, more QTP sessions may share a hop (or path), potentially leading to more competition for limited quantum memory resources and this results in a smaller sending windows. Such a tradeoff means that there is no clear correlation between the network size and sending window size.

It is interesting to note, however, that sessions can achieve the largest sending window size under \tele, while FRA results in the smallest average sending window sizes. To see why, consider two sessions, $S1$ and $S2$ which go through the same repeater $R1$ with $2C$ units of quantum memory (thus a capacity of $C$ data qubits) throughout the lifetime of $S1$. Further assume that $S2$ shares a bottleneck node $R2$ with capacity $C$ with ten other sessions throughout its lifetime. In EW, although the sending window size of $S2$ will never be more than $C/10$ due to the bottleneck at $R2$, the sending window size of $S1$ will always be limited to $C/2$. However, in \tele, $S1$ may grow its sending window size up to $9C/10$ (as long as there is no other bottleneck for $S1$). The only potential downside in \tele is that when a CE flag is set, a sending window size $W$ will get reduced by half, even though there may be a few ($< W/2$) units of quantum memory are available which may get wasted. For similar reasons, FRA results in the smallest average sending window size as it combines the worst of EW and \tele.

It can also be seen by comparing Fig.~\ref{subfig:nodeWinTR} with other subfigures that the effective window size in \mhqdn is even lower than that using FRA. This is mainly due to the following two facts: (i). For each \plain session, different hops along the path lack coordination in that while the sending window for the session may be high on a previous hop, it is low on a later hop, and the effective window size is limited by the smallest one along the path; and  (ii). The amount of quantum memory used for receiving purposes will limit the maximum window size. In a \mhqdn, only 4/13 of the quantum memory is used for receiving (sharings), while 1/3 of the quantum memory are used for receiving (data qubits) in a \tpqdn. 

In addition, we observe that although some of the QTP sessions may be routed through a lightly loaded path and achieve a very large average sending window, most of the sessions achieve the similar average sending window size. This indicates that all these protocols has a good fairness performance. For example, if we ignore 10\% of the sessions with the largest average sending window size (because they are routed to a lightly loaded path), the Jain's fairness index with all four QTP cases are 0.9262, 0.9851, 0.9857, and 0.8784, respectively, when there are 40 repeaters.  

\begin{figure*}
		\subfigure[\tpqdn with \tele.]{\label{subfig:nodeUtiTele}
		\includegraphics[width=0.23\linewidth]{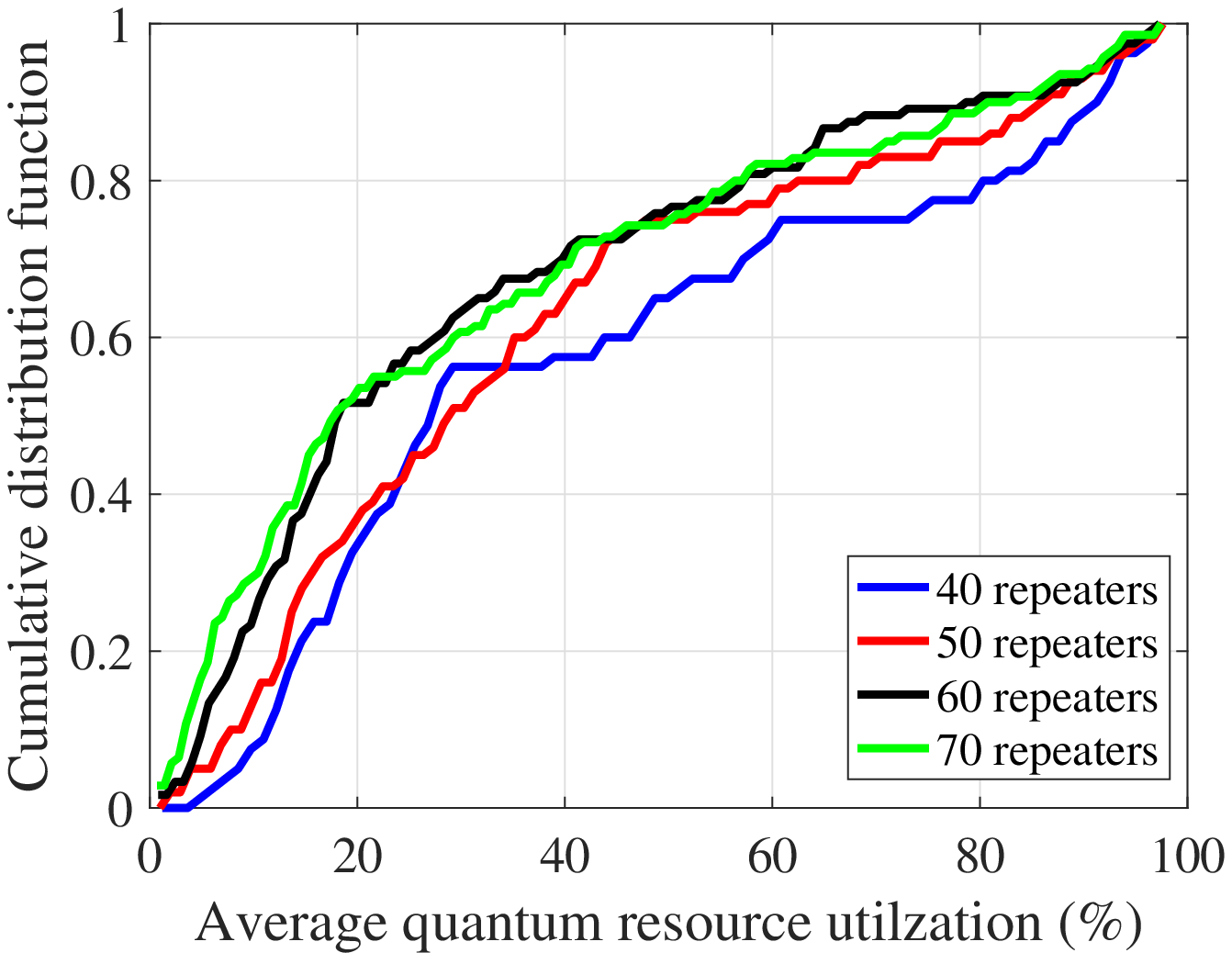}
	}
		\subfigure[\tpqdn with EW.]{\label{subfig:nodeUtiExplict}
		\includegraphics[width=0.23\linewidth]{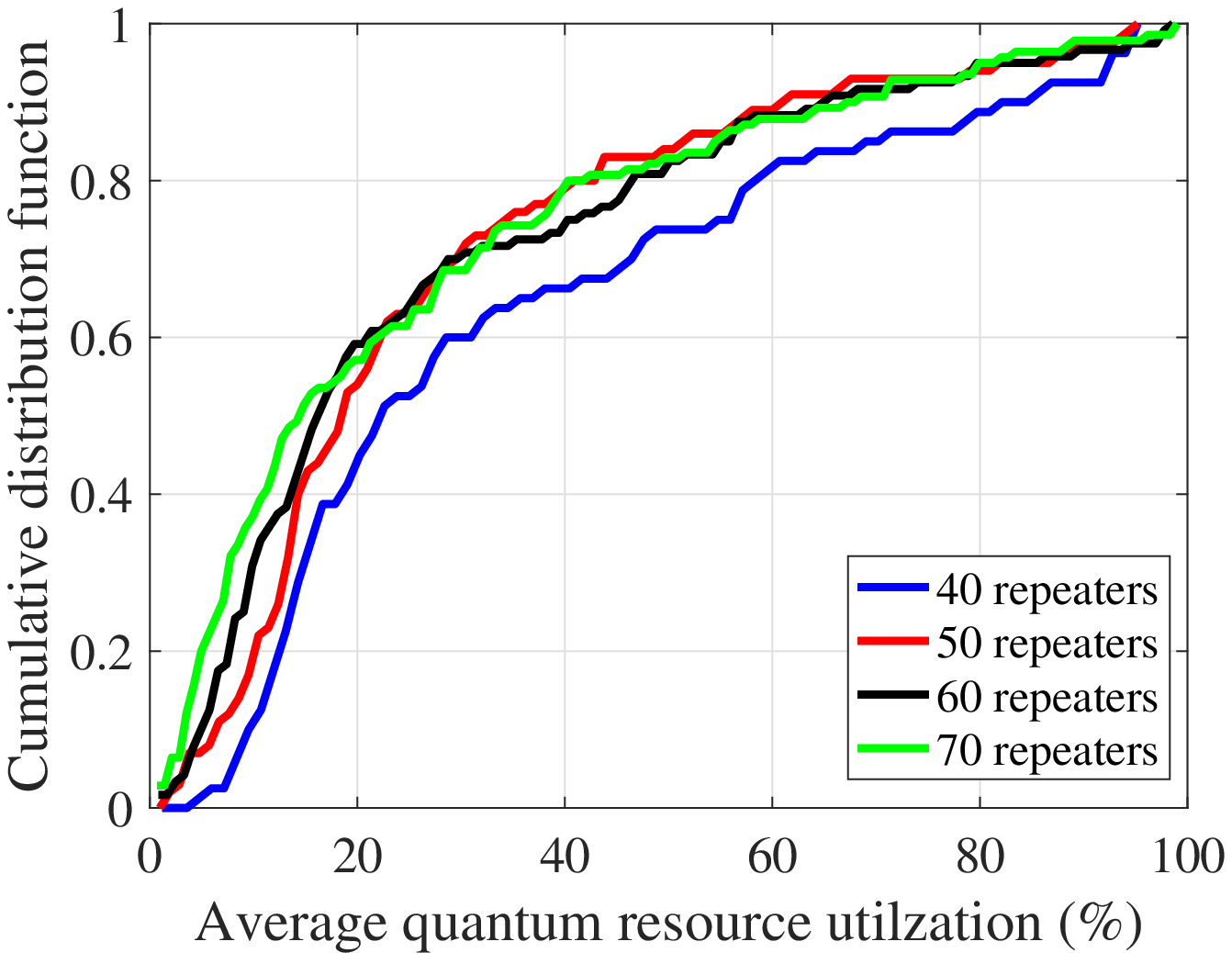}
	}
	\subfigure[\tpqdn with FRA.]{\label{subfig:nodeUtiFair}
		\includegraphics[width=0.23\linewidth]{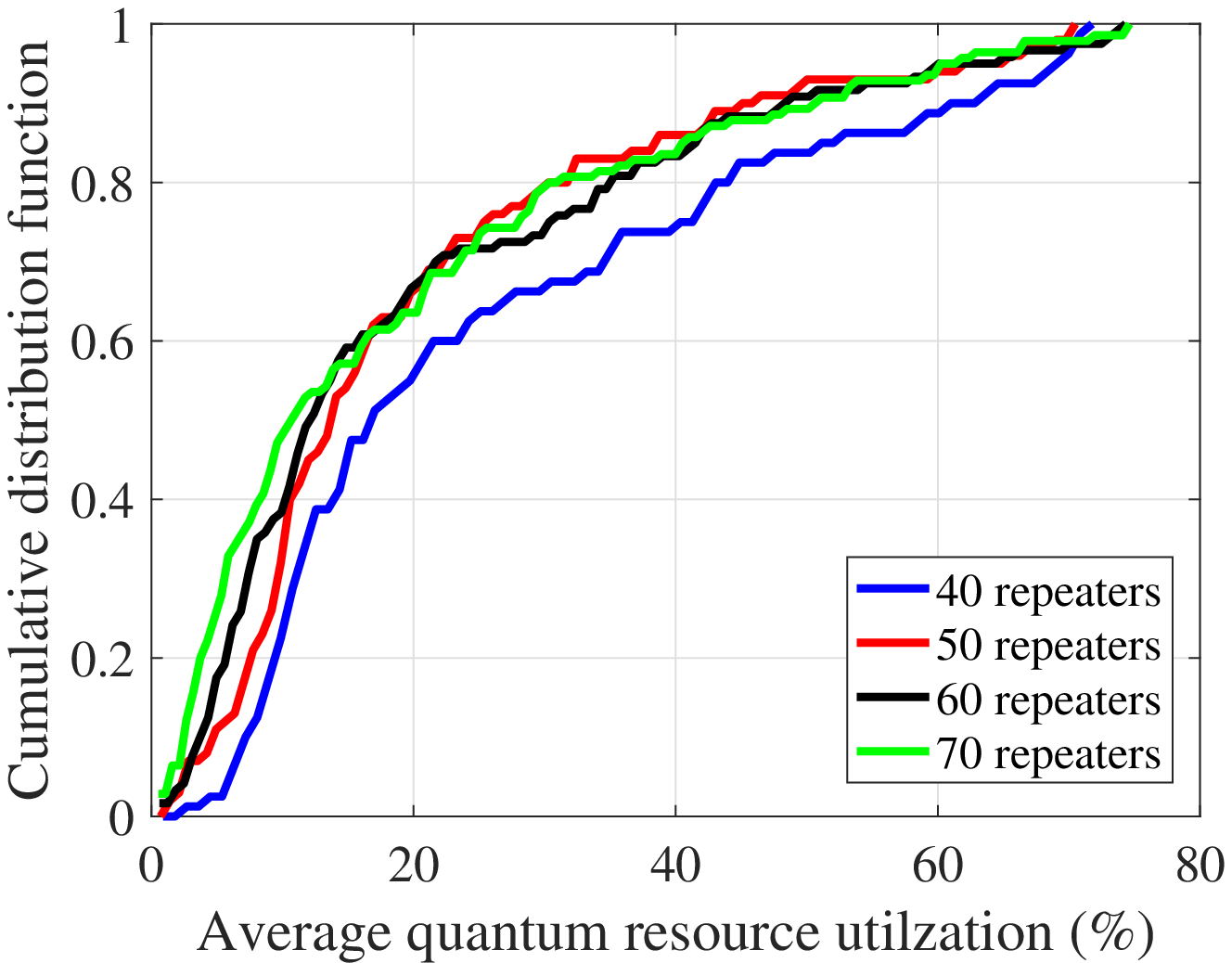}
	}
	\subfigure[\mhqdn with \plain.]{\label{subfig:nodeUtiTR}
		\includegraphics[width=0.23\linewidth]{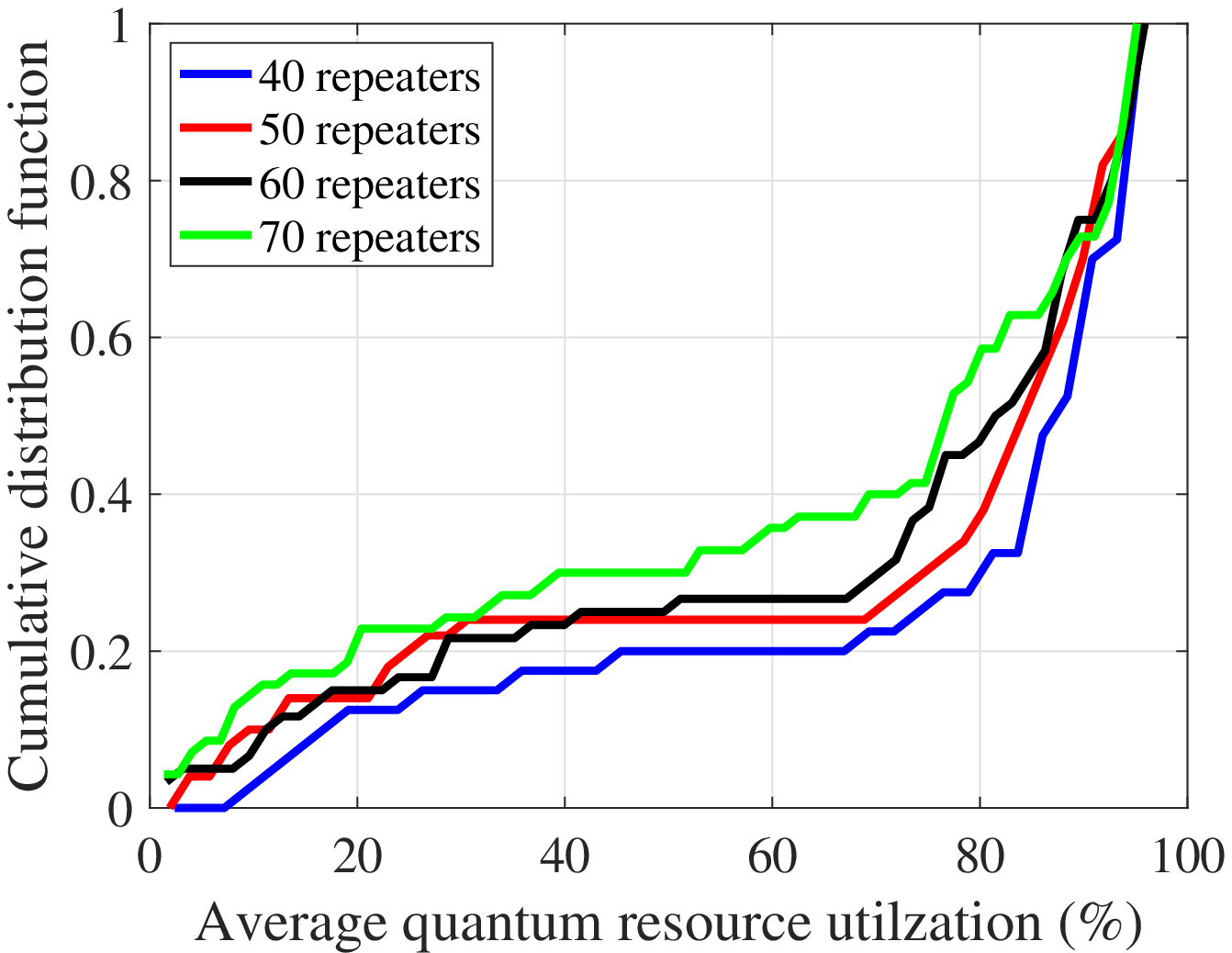}
	}
	\vspace{-0.1in}
	\caption{How network size impacts quantum memory utilization.}\label{fig:sizeUtilization}
	\vspace{-0.2in}
\end{figure*}

We observe that from Fig.~\ref{fig:sizeWin}, regardless of which QTP is adopted, a smaller network size tends to results in a higher quantum memory utilization. This is very intuitive since each node in a smaller network will carry  more sessions. Nevertheless, we should also note that the number of nodes whose quantum memory utilization is close to 1 does not significantly increase in a small network. 
This is because the QTP sessions tend to use nodes with a lower quantum memory utilization under a load balanced routing protocol.

Comparing all four subfigures in Fig.~\ref{fig:sizeUtilization}, we can see that as expected, \tele leads to a higher quantum memory utilization than EW, and FRA results in the lowest average quantum memory utilization among the three, for the same reason as that discussed for the case about the sending window size. However, it may come as a surprise that \plain in a \mhqdn leads to the highest quantum memory utilization (even though it has the lowest effective sending window size). \plain achieves this because  every relay  can maximize its own resource utilization without having to worry about any potential bottleneck along the end-to-end path from Alice to Bob. However, as we will see in Fig.~\ref{fig:throuNode}, despite this, \plain in a \mhqdn cannot achieve a good throughput. While one of the reasons is that data qubits have to go through multiple hops,
 we will discuss the other major reason below.


\begin{figure}
	\begin{minipage}{0.48\linewidth}
		\centering
		\includegraphics[width=0.85\linewidth]{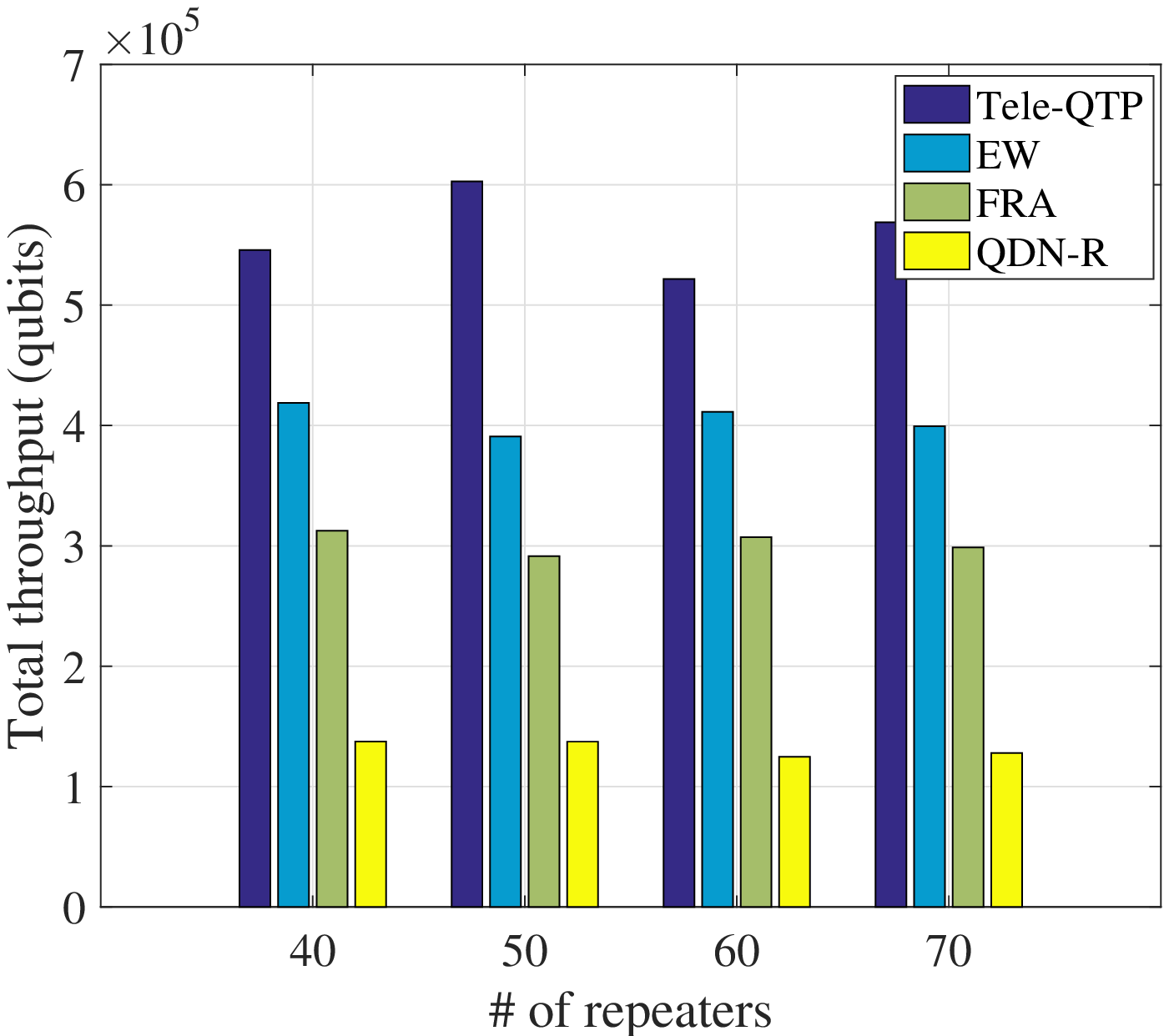}
		\vspace{-0.1in}
		\caption{Network size vs. throughput.}\label{fig:throuNode}
	\end{minipage}
	\begin{minipage}{0.48\linewidth}
		\centering
		\includegraphics[width=0.85\linewidth]{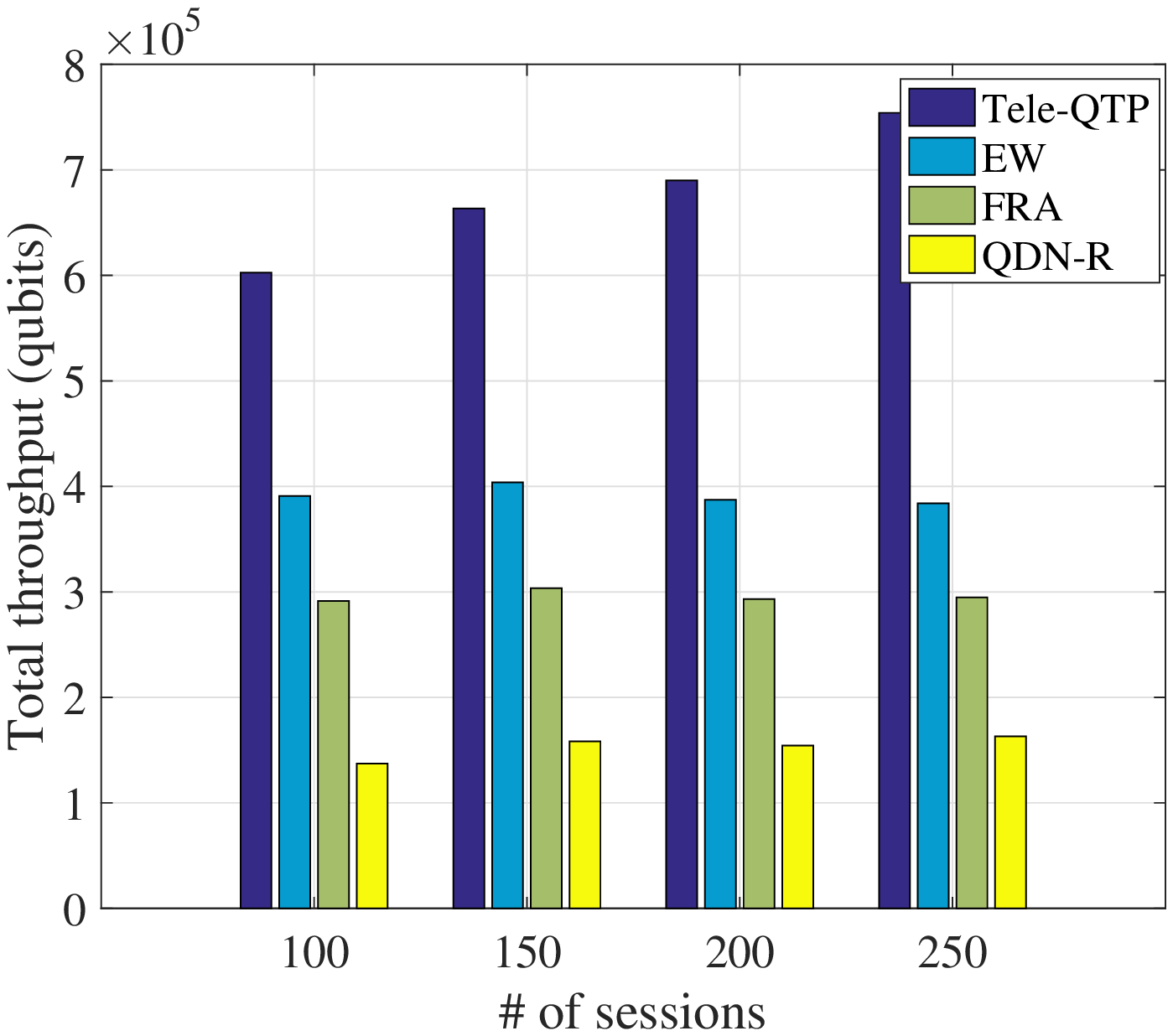}
		\vspace{-0.1in}
		\caption{Session number vs. throughput.}\label{fig:throuConn}
	\end{minipage}
	\vspace{-0.2in}
\end{figure}
From Fig.~\ref{fig:throuNode}, we observe that throughput does not have a clear correlation with the network size, for the same reason discussed for the case about the sending window size. However, it is worth noting that regardless of the network size, \tele achieves the largest throughput. When there are 50 repeaters in the network for example, \tele significantly outperforms EW, FRA, and \plain by 54.22\% and 106.84\%, and 290.12\%, respectively. Besides the fact that \mhqdn is a multi-hop network, we believe the main reason for the extremely poor throughput performance of \plain is the assumption that we need to transmit 2 sharings for each data qubit, in order to achieve reliable delivery.


\noindent\textbf{Effect of workload:}
To investigate the effect of workload on the performance of proposed QTPs, we assume that there are 50 repeaters, and the average node degree among repeaters is 4. In such a network, we vary the number of QTP sessions from 150 to 250 (in increment of 50) injected into the network, and collect the performance of different QTPs. 

From Fig.~\ref{fig:throuConn}, we can see that the network throughput increases with the number of sessions  when \tele is adopted, indicating that \tele can utilize the  quantum memory quite effectively. However, the throughput achieved by other three QTPs is capped low. Worse,  more sessions may even reduce the throughput in EW or FRA (\eg  with 200 vs 150 sessions) due to their inefficiencies in allocating quantum memory (and managing the sending window size). Our simulations show that \tele outperforms EW, FRA, and \plain by 96.37\%, 155.85\% and 310.94\%, respectively, when there are 250 QTP sessions in the network.

\begin{figure*}
	\subfigure[\tpqdn with \tele.]{\label{subfig:conWinTele}
		\includegraphics[width=0.23\linewidth]{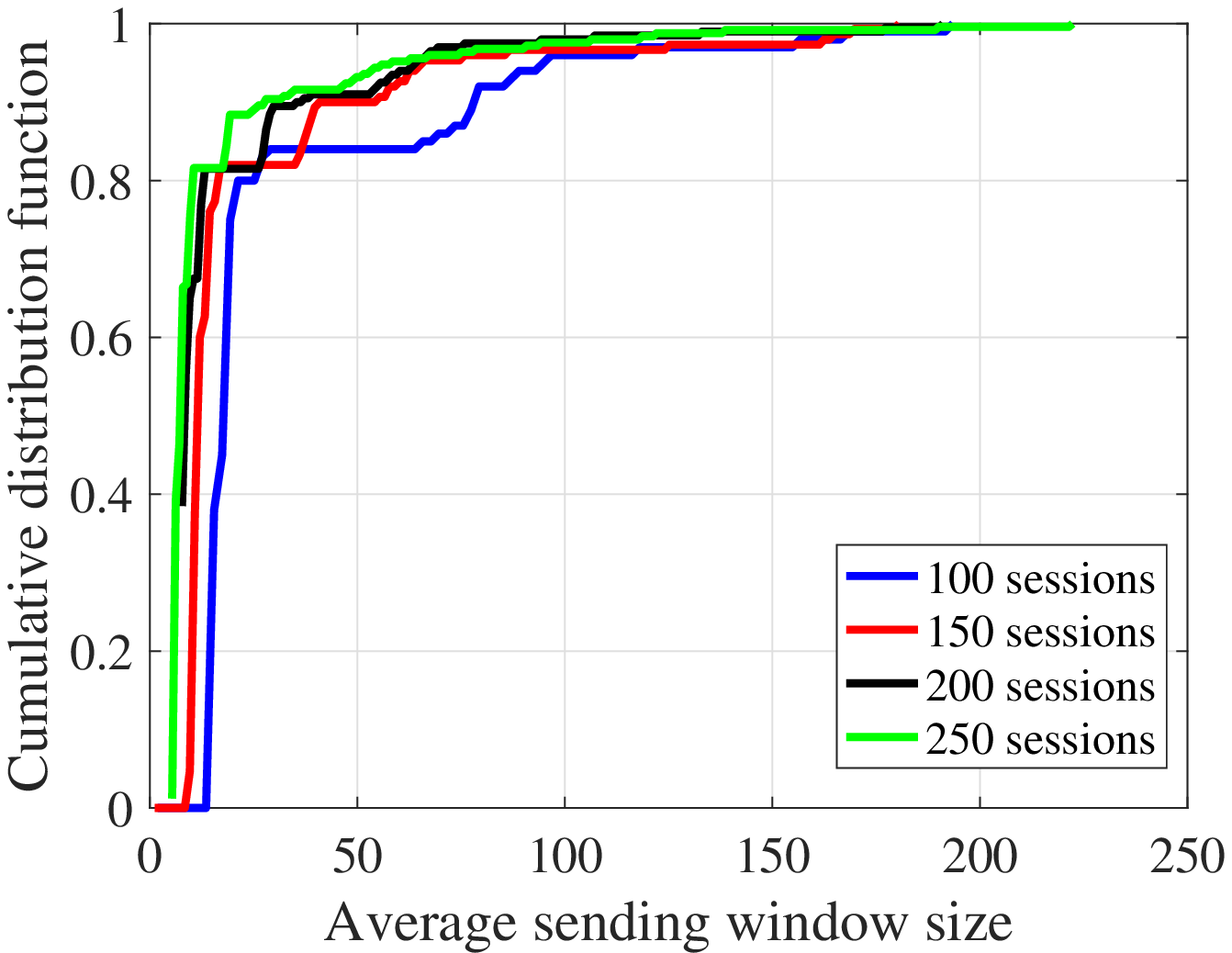}
	}
	\subfigure[\tpqdn with EW.]{\label{subfig:conWinExplict}
		\includegraphics[width=0.23\linewidth]{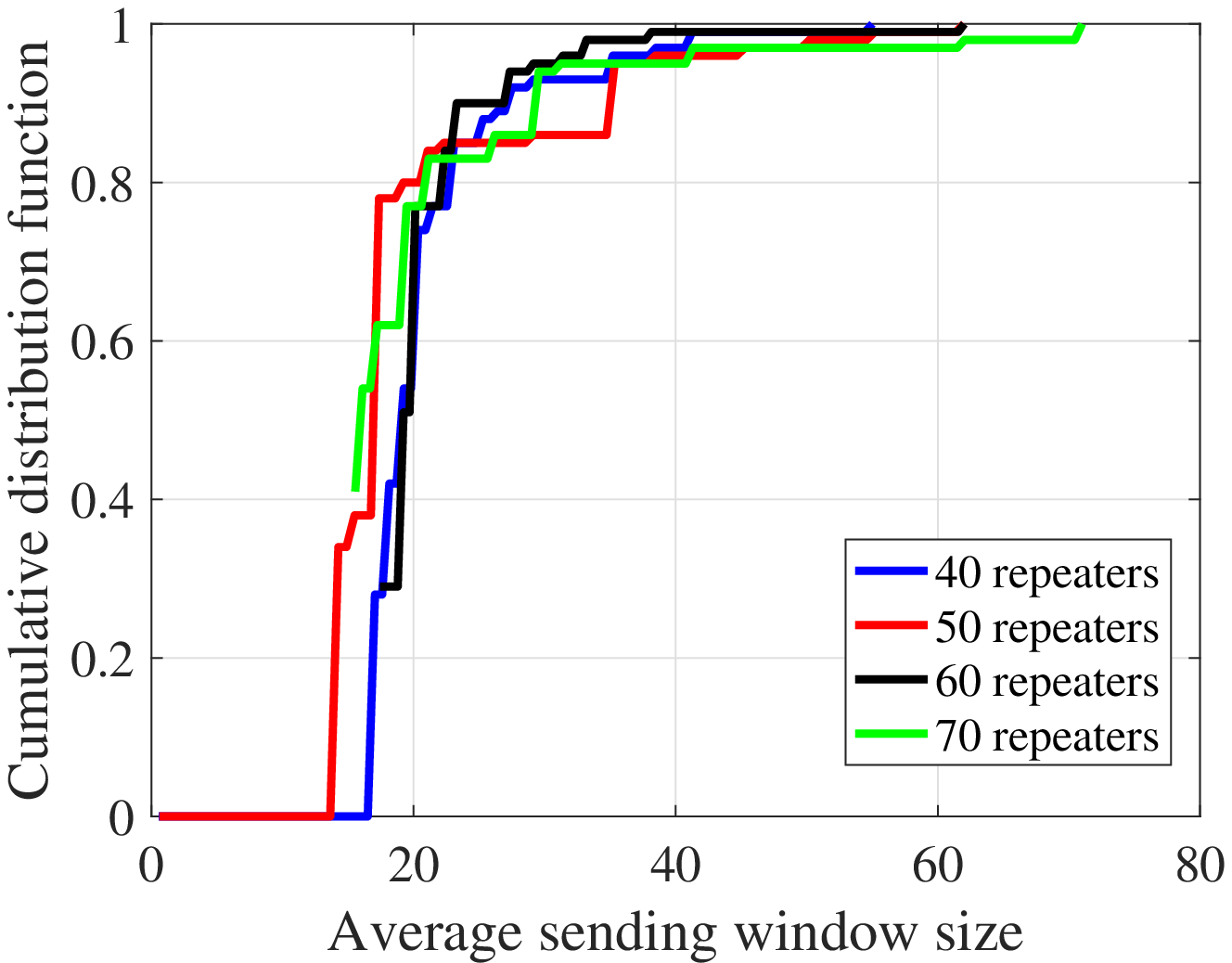}
	}
	\subfigure[\tpqdn with FRA.]{\label{subfig:conWinFair}
		\includegraphics[width=0.23\linewidth]{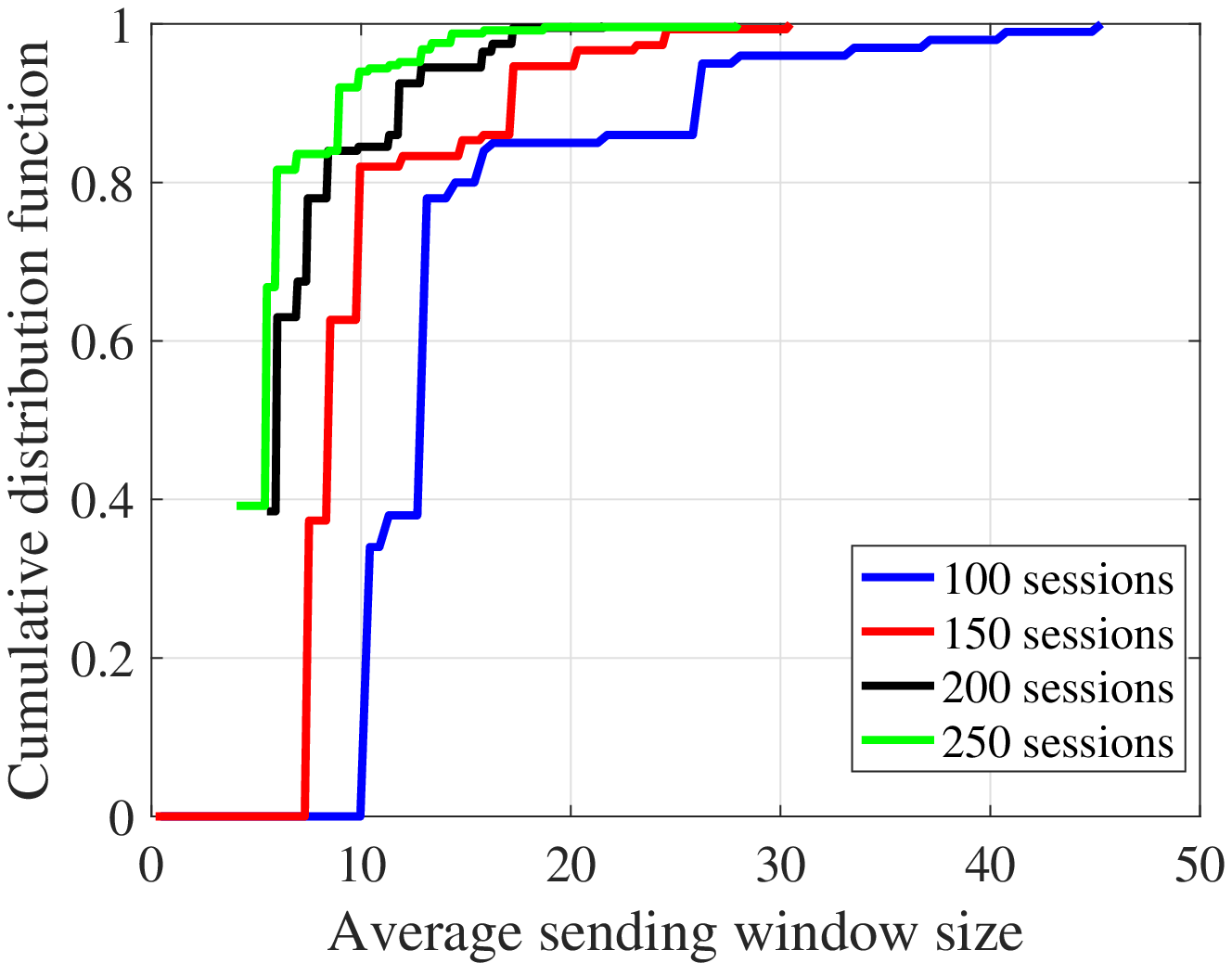}
	}
	\subfigure[\mhqdn with \plain.]{\label{subfig:conWinTR}
		\includegraphics[width=0.23\linewidth]{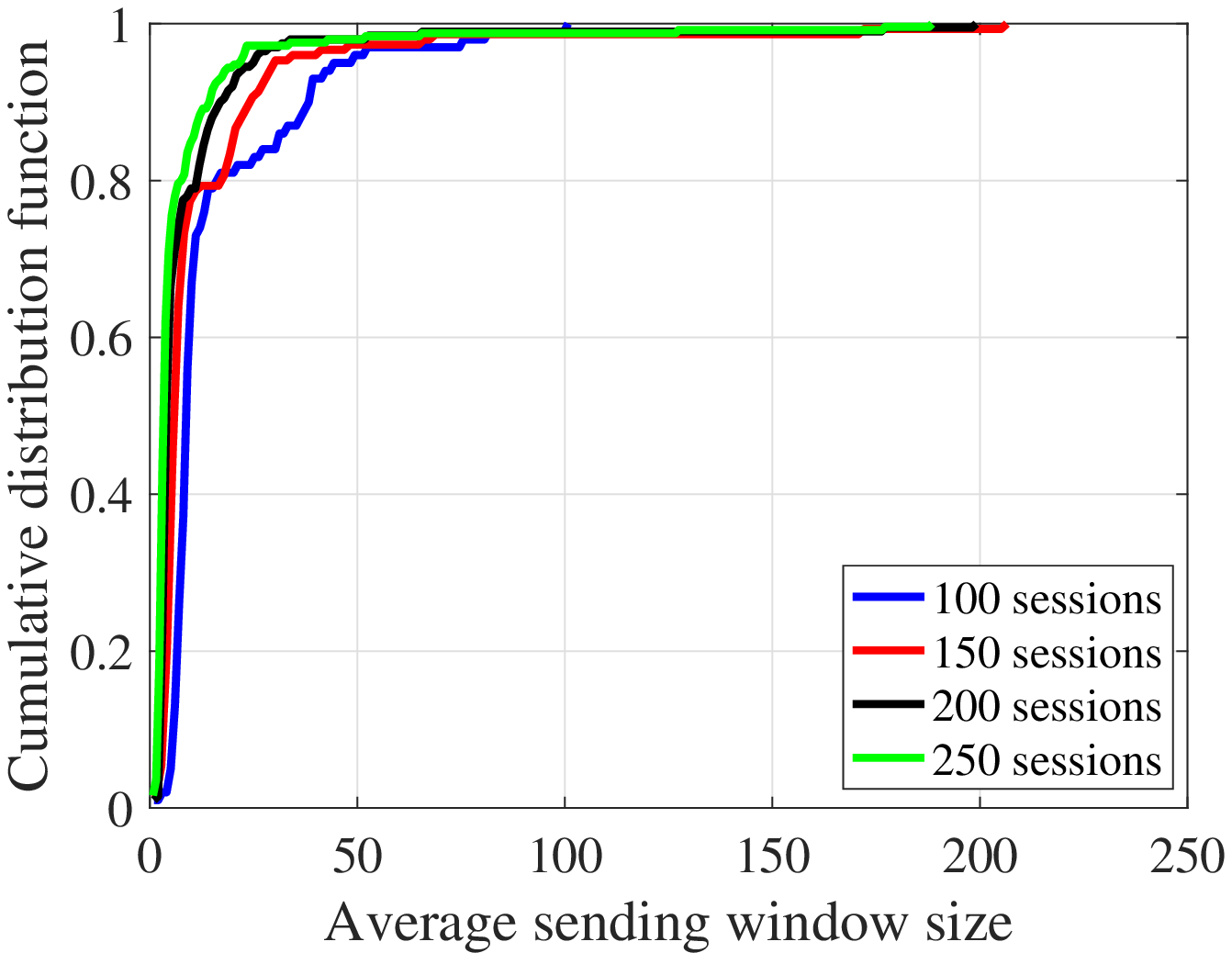}
	}
	\vspace{-0.1in}
	\caption{How number of sessions impacts sending window sizes.}\label{fig:conWin}
	\vspace{-0.2in}
\end{figure*}

How the sending window size varies with the workload is presented in Fig.~\ref{fig:conWin}. As expected, more sessions in a network result in a smaller average sending window size as more sessions are sharing the same amount of quantum memory. From this figure, we can also observe that no matter which QTP is adopted, most of the sessions will be able to achieve a similar average window size. For example, if we consider the bottom 90\% of the total 250 sessions (in terms of sending window sizes), the Jain's fairness index values achieved by the four QTPs are 0.9993, 0.9598, 1.0000, and 0.9300, respectively. In addition, as in the case with a varying network size , \tele achieves the largest average window sizes among the three that use \tp for quantum data exchanges, since any \tele session can use the idle memory without considering if it obtains the amount of memory above average.

\begin{figure*}
	\subfigure[\tpqdn with \tele.]{\label{subfig:conUtiTele}
		\includegraphics[width=0.23\linewidth]{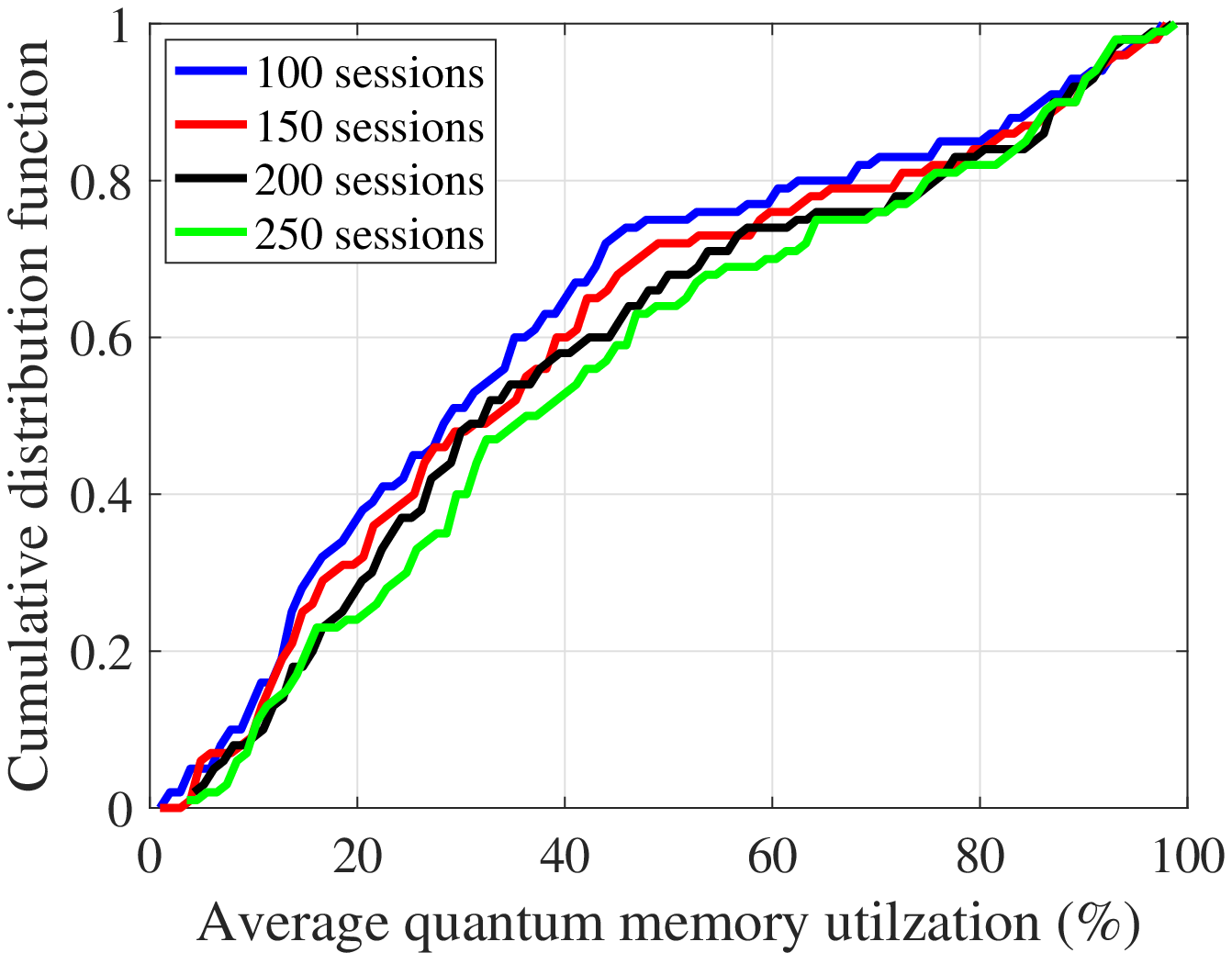}
	}
	\subfigure[\tpqdn with EW.]{\label{subfig:conUtExplict}
		\includegraphics[width=0.23\linewidth]{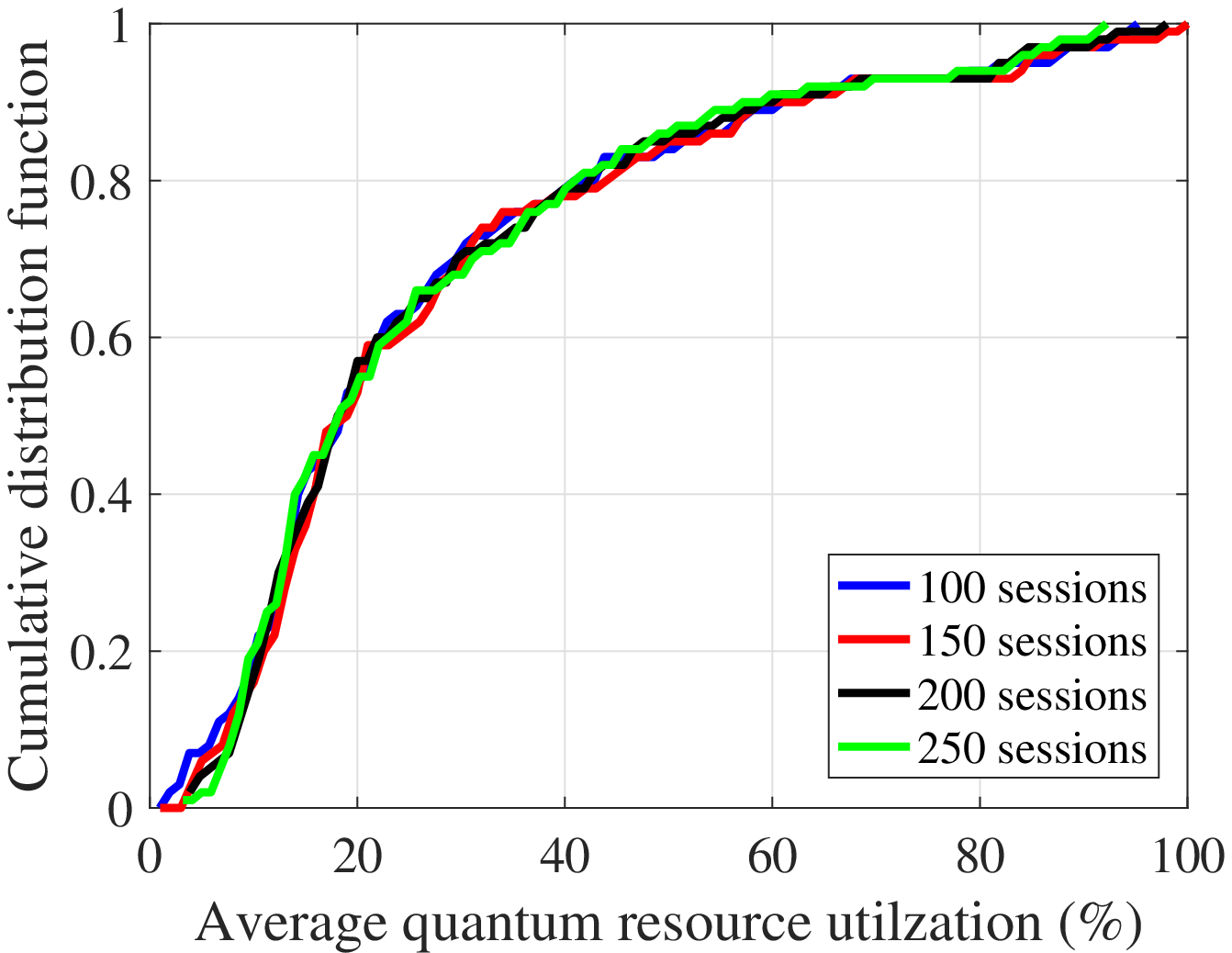}
	}
	\subfigure[\tpqdn with FRA.]{\label{subfig:conUtiFair}
		\includegraphics[width=0.23\linewidth]{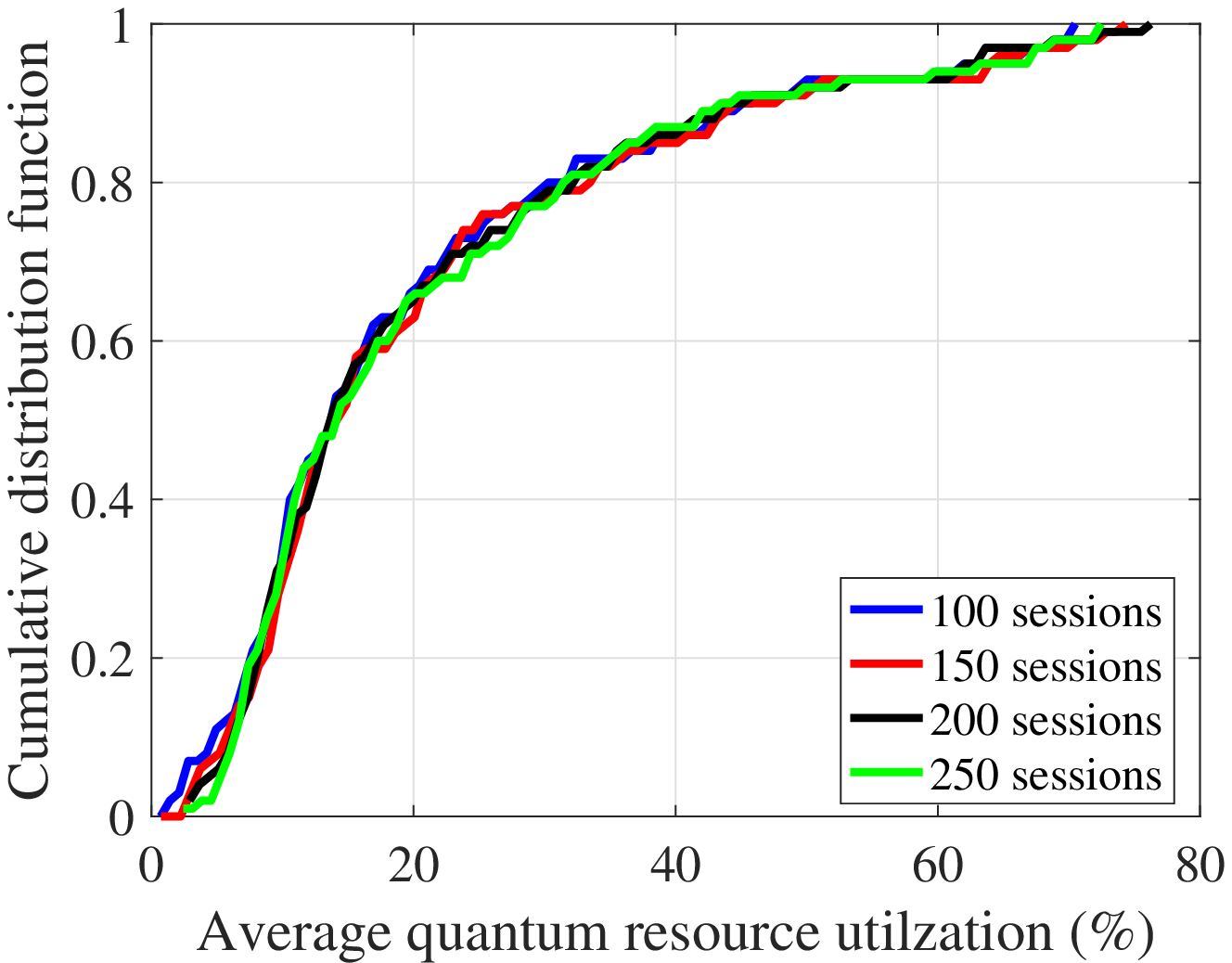}
	}
	\subfigure[\mhqdn with \plain.]{\label{subfig:conUtiTR}
		\includegraphics[width=0.23\linewidth]{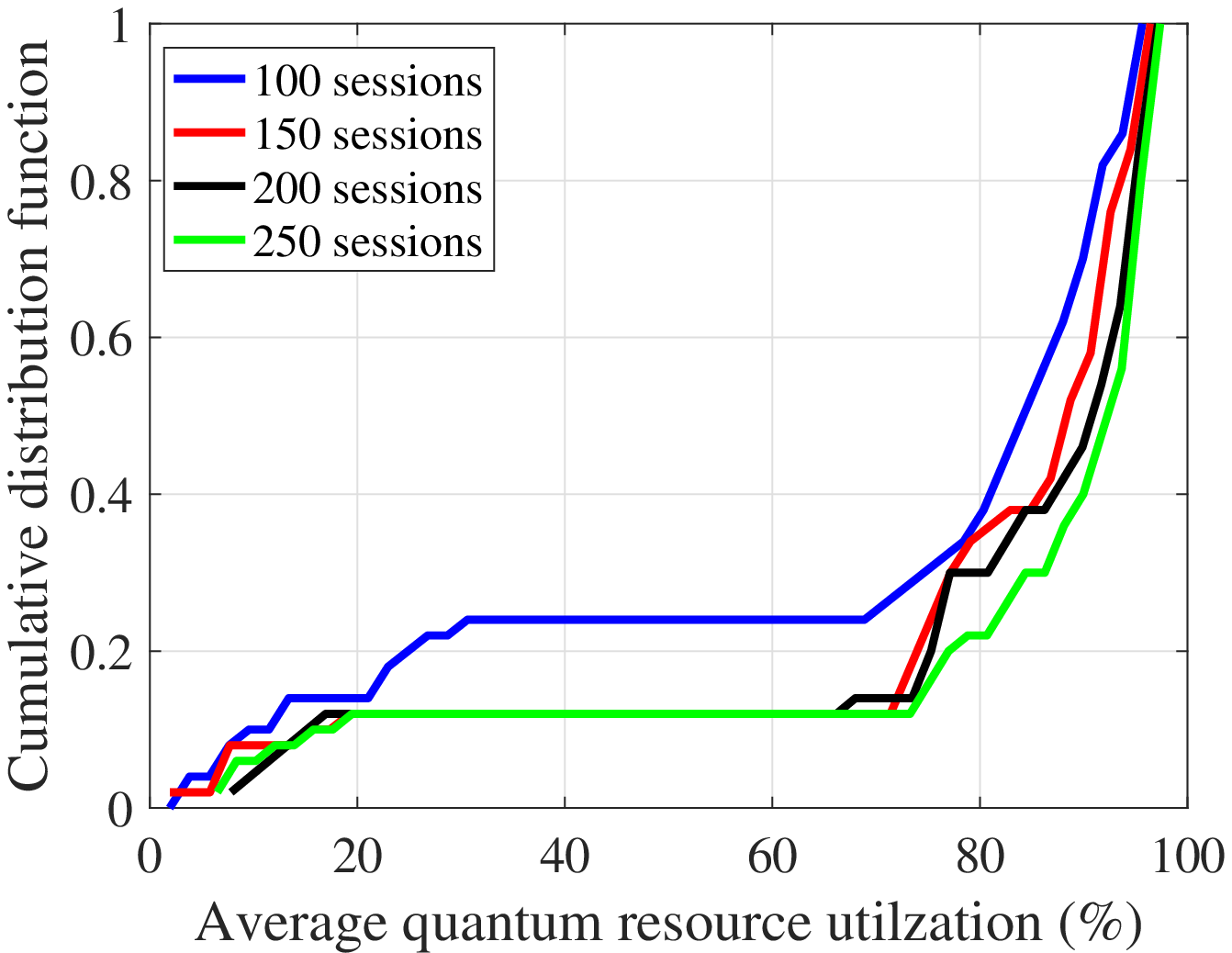}
	}
	\vspace{-0.1in}
	\caption{How number of sessions impacts quantum resource utilization.}\label{fig:conUtilization}
	\vspace{-0.2in}
\end{figure*}

Normally, we would expect that the quantum memory utilization will increase with the number of sessions in the network. 
From Fig.~\ref{fig:conUtilization}, we see that this is true in \tele and \plain. However, with either EW or FRA, when a larger number of sessions going through a node, the maximal sending window size of each session becomes smaller. Due to the presence of bottleneck and inefficiency in their ways of allocating quantum memory among multiple competing sessions, the quantum memory utilization in EW and FRA won't necessarily increase.



\subsection{Tradeoffs between \tpqdn and \shqdn}\label{subsec:tradeoff}
\begin{figure}
	\subfigure[Impact of sending success probability in \shqdn.]{\label{subfig:prob}
		\includegraphics[width=0.45\linewidth]{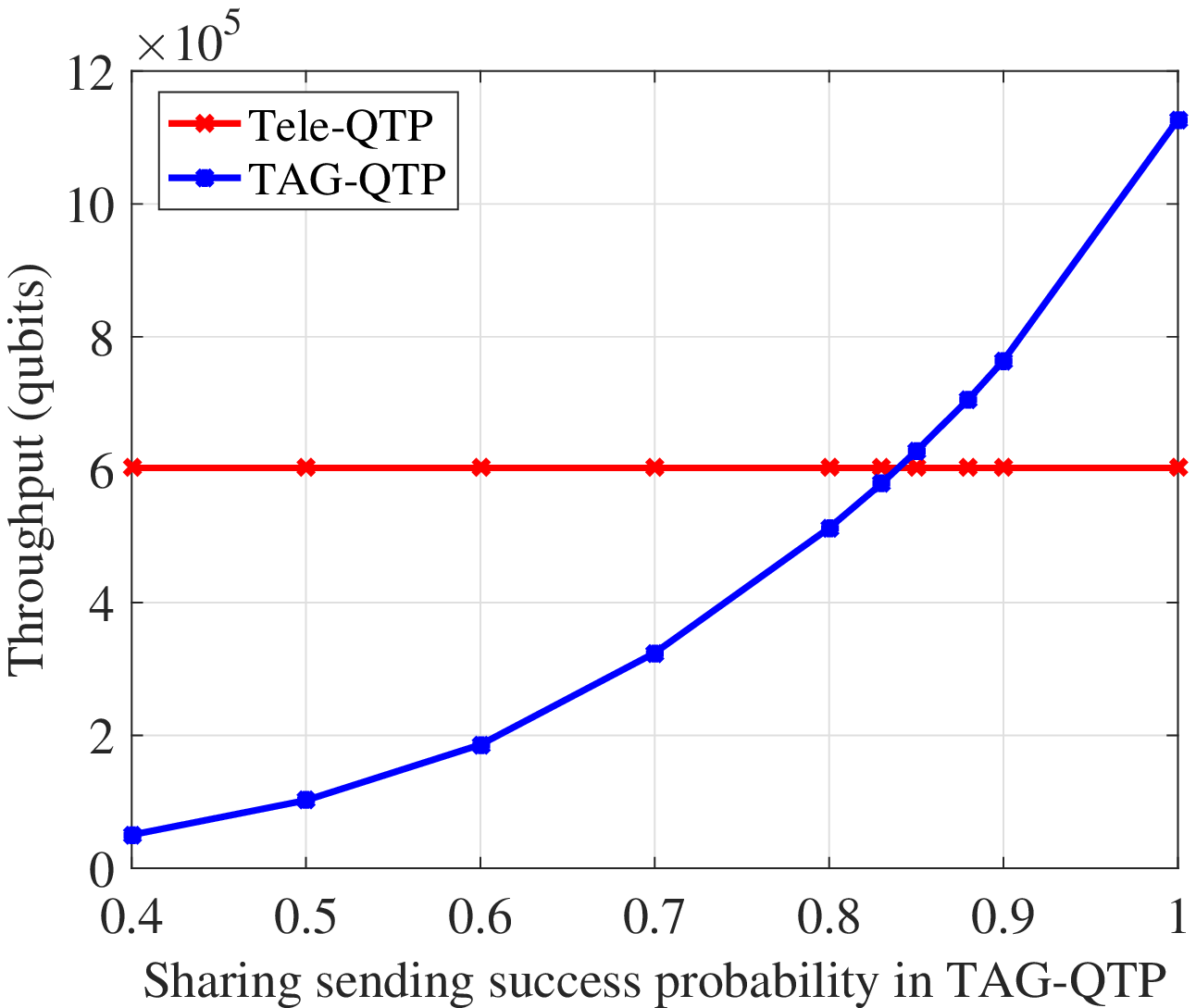}
	}
	\subfigure[Impact of time slot length.]{\label{subfig:slot}
		\includegraphics[width=0.45\linewidth]{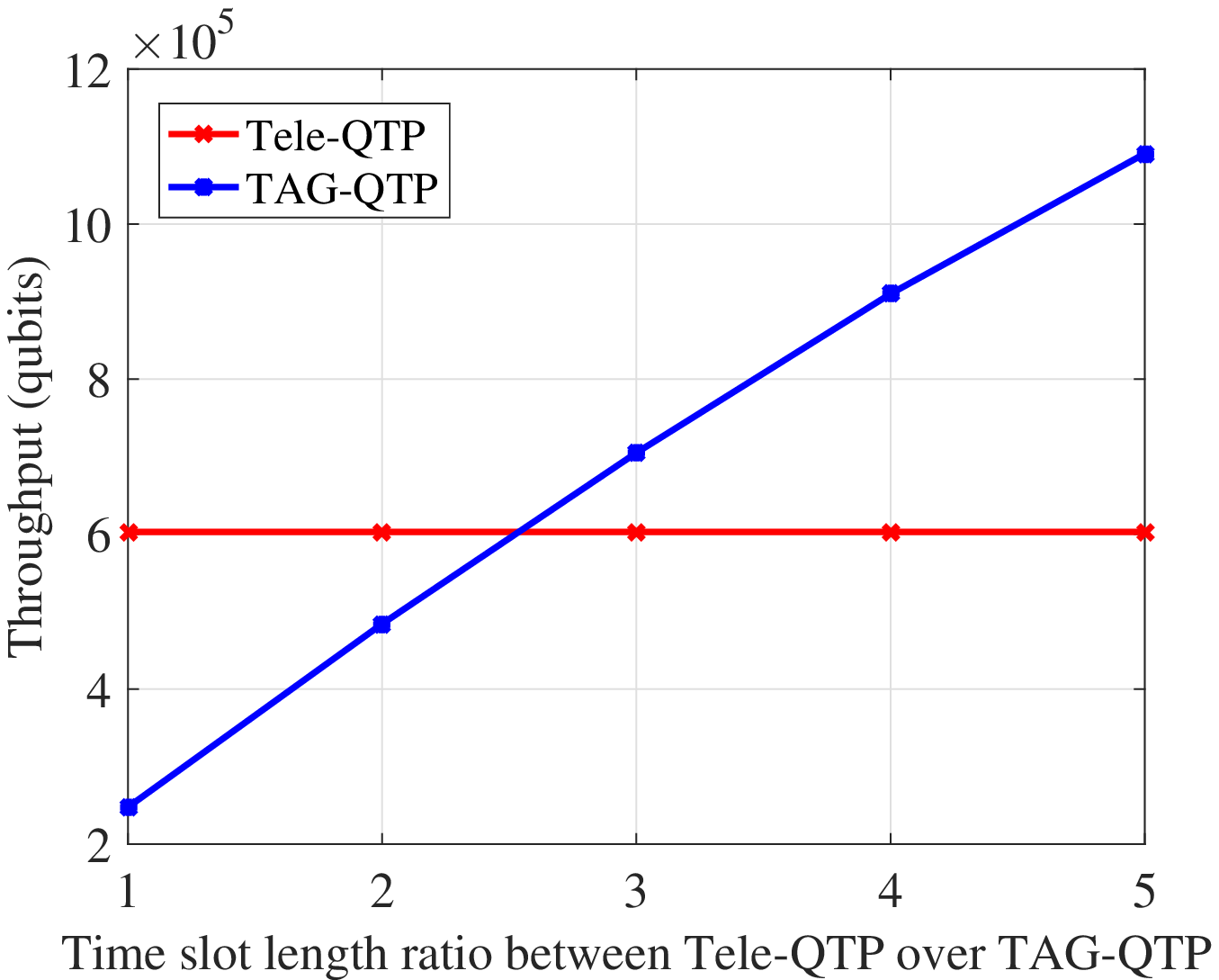}
	}
	\vspace{-0.1in}
	\caption{\tele vs. \plain in \shqdn.}\label{fig:tradeoff}
	\vspace{-0.2in}
\end{figure}

So far, we have only considered the case where \plain is deployed in a \mhqdn. Technology aside and  performance wide, it seems that under our assumptions made so far, a \mhqdn does not have much advantage over a \tpqdn. 
In this subsection, we seek to answer a blue-sky question of if and how a \shqdn might be a better choice. If we were to ignore the effect of the transmission distance on the  transmission loss probability, the performance of a \shqdn is an upper bound on that of a \mhqdn. Accordingly, comparing the performance of a \shqdn and a \tpqdn may further shed light on the comparison we have done so far between a \tpqdn and \mhqdn as well.  Accordingly, we simulate  a \shqdn with 50 all-optical switches or a  \tpqdn with 50 repeaters,  and study their throughput after injecting 100 sessions into both networks.

In the simulations,  we consider two main parameters impacting the throughput: the length of each time slot and the success probability to deliver each sharing. The former can be much smaller in a \shqdn (or \mhqdn) than that in a \tpqdn as explained earlier. In this study, we assume that a time slot in a \tpqdn can be 1 to 5 times longer.


To study the effect of the success probability of sending each sharing, we first assume that the length of one time slot is the same in both \tpqdn and \shqdn, and change the success probability. 
As a futuristic study, we will not limit the success probability to a small value afforded by the state-of-the-art technology~\cite{singlephoton}. This is because today's transmission success probability is too low for supporting quantum data exchanges, although it is enough for some other applications, such as QKD. Instead, we will assume a high success probability, which may be achieved with the advance in technology and/or in a room size network for clustering several small quantum computers. 

From Fig.~\ref{subfig:prob}, we can see that when each sharing can be sent successfully with a probability close to 1, \shqdn can indeed achieve a higher throughput than \tpqdn, despite the need to encode each data qubit with multiple sharings. This is because that without having to share (or compete against) with each other the quantum memory at repeaters as in a \tpqdn, each session can achieve a larger sending window, which leads to a larger throughput. When each sharing can be delivered with a probability about 0.84, these two QDNs can achieve almost the same network throughput. 

To study the effect of the length of each time slot in the \shqdn on the performance, we fix the success probability in the \shqdn to 0.65.  From  Fig.~\ref{subfig:slot}, we can see that a shorter time slot in the \shqdn leads to a larger throughput in almost liner fashion. When the time slot length in the \tpqdn is about 2.5x of that in the \shqdn, both methods achieve almost the same network throughput. 

\vspace{-0.1in}
\section{Conclusions}\label{sec:conclu}

In this work, two first-of-the-kind quantum transport protocols, \tele and \plain, have been proposed which use novel mechanisms to support reliable quantum data exchanges in quantum data networks (QDNs). Different from the classical TCP which assumes that payloads can be sent as packets, and lost packets can be retransmitted using their copies in a sending buffer, both \tele and \plain will have to deal with streams of data qubits, avoid congestion in a QDN, and  manage the limited quantum memory at not only Alice and Bob, but also quantum nodes (repeaters and relays). Our analysis and extensive simulations have shown that \tele and \plain can achieve not only a high throughput, but also fair (and efficient) resource allocation for all the QTP sessions sharing the same bottleneck. Comparisons between the two methods for quantum data exchanges, namely \tp and \plain and their corresponding QDNs have been made and the results help shed new light on the tradeoffs involved. There exist many extension or improvement works at the transport layer and other network layers, including distributed quantum computing algorithms that can take advantage of the proposed QTPs for QDNs and the enabled cluster of multiple small quantum computers.

\bibliographystyle{ACM-Reference-Format}
\bibliography{main}

\appendix

\section{State Machine for sending a data qubit with (2,3)--threshold sharing scheme}\label{app:sharing}
\begin{figure}
	\centering
	\includegraphics[width=0.8\linewidth]{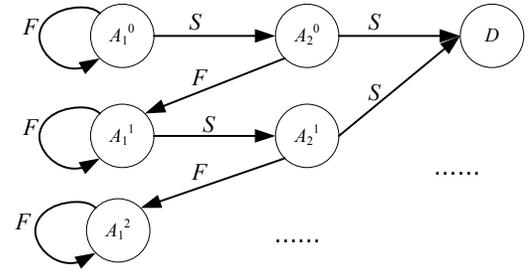}
	\vspace{-0.1in}
	\caption{State machine to transmit data qubit $A$ with encoded sharings.}\label{fig:coding}
	\vspace{-0.2in}
\end{figure}
In Fig.~\ref{fig:coding}, we show the state machine to send a data qubit $A$ by leveraging the quantum secret sharing, where $F$ and $S$ are short for "Fail" and "Succeed", respectively. In this state machine, $A_1^0A_2^0A_3^0$ are the (2,3)-threshold sharings of the original data qubit $A$, and $A_1^kA_2^kA_3^k$ are the (2,3)-threshold sharings of $A_3^{k-1}$ for $k\geq 1$. In state $A_n^k$, sharing $A_n^k$ is to be sent, action $F$ means sending fails, action $S$ means sending successes, and the data qubit is successfully delivered in state $D$.

\section{Proof of Theorem~\ref{the:congestion}}\label{app:congestion}
\begin{proof}
	In a steady state, all the QTP sessions are in the CA state. Let $W^{(t)}_n$ be the window size of session $n$ at time slot $t$, then each QTP session $n$ will announce its window size as $W^{(t+1)}_n = W^{(t)}_n + 1$. If there is no congestion at time slot $t$, we know $\sum_n W^{(t)}_n \leq C$.
	According to Algorithm~QMA, at time slot $t+1$, the sending window of session $n$ can be cut to $\lfloor (W^{(t)}_n+1)/2 \rfloor$.
	Since $\lfloor (W^{(t)}_n+1)/2 \rfloor \leq W^{(t)}_n$, we know $\sum_n \lfloor W^{(t)}_n+1 \rfloor \leq \sum_n W^{(t)}_n \leq C$, \ie, if there is no congestion at time slot $t$, so will be at time slot $t+1$. Now, we consider the starting time slot. The window sizes of all the QTP sessions are all 1, and they require quantum memory to support transmitting $N$ qubits in total. Based on the assumption that $N\leq C$, all the qubits can be served. Accordingly, we can conclude that our proposed QTPs is congestion-free.
\end{proof}

\section{Proof of Theorem~\ref{the:fairness}}\label{app:fairness}
\begin{figure*}[htbp]
	\subfigure[Sending window size in \tele.]{\label{subfig:microTele}
		\includegraphics[width=0.31\linewidth]{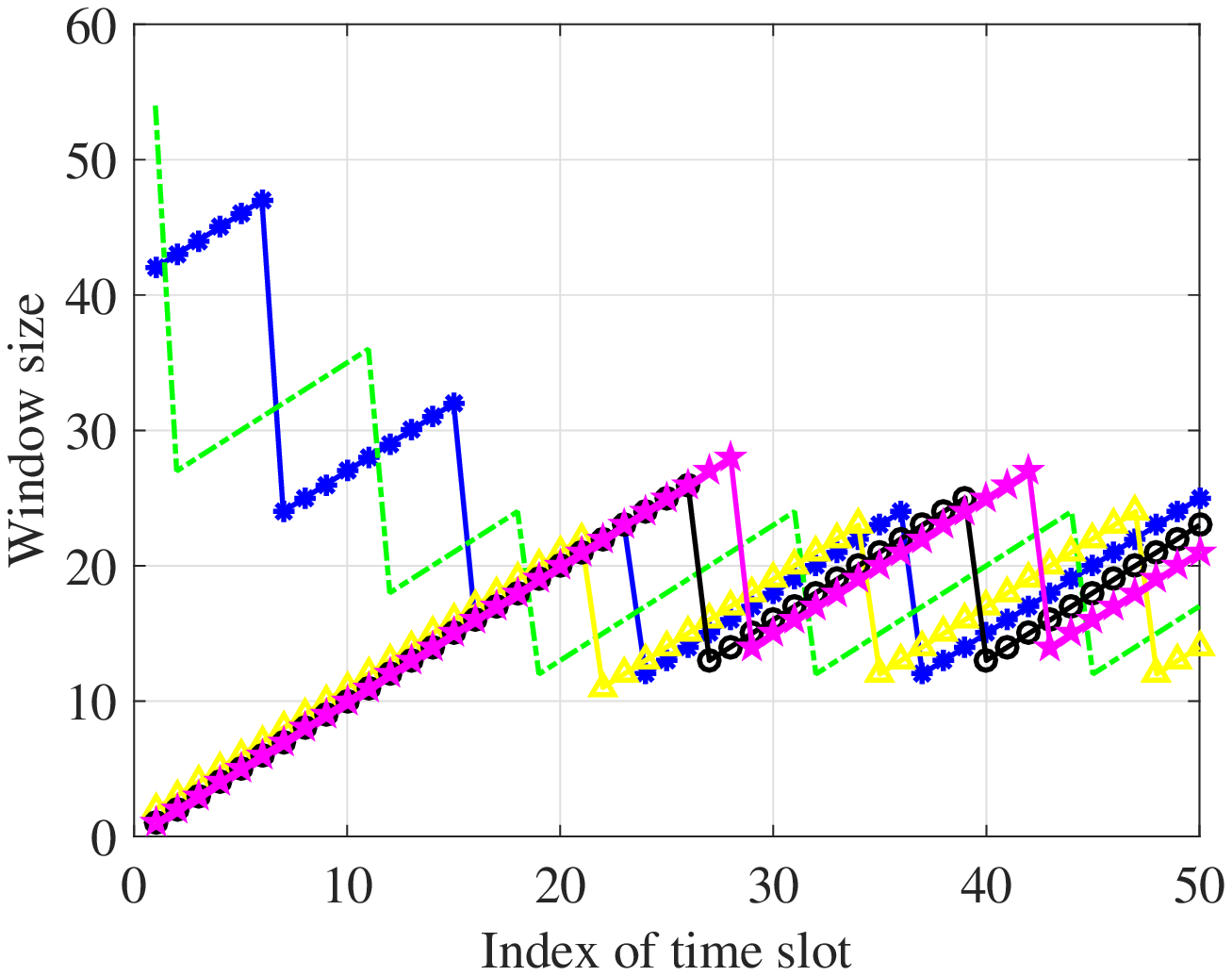}
	}
	\subfigure[Sending window size in \plain.]{\label{subfig:microPlain}
		\includegraphics[width=0.31\linewidth]{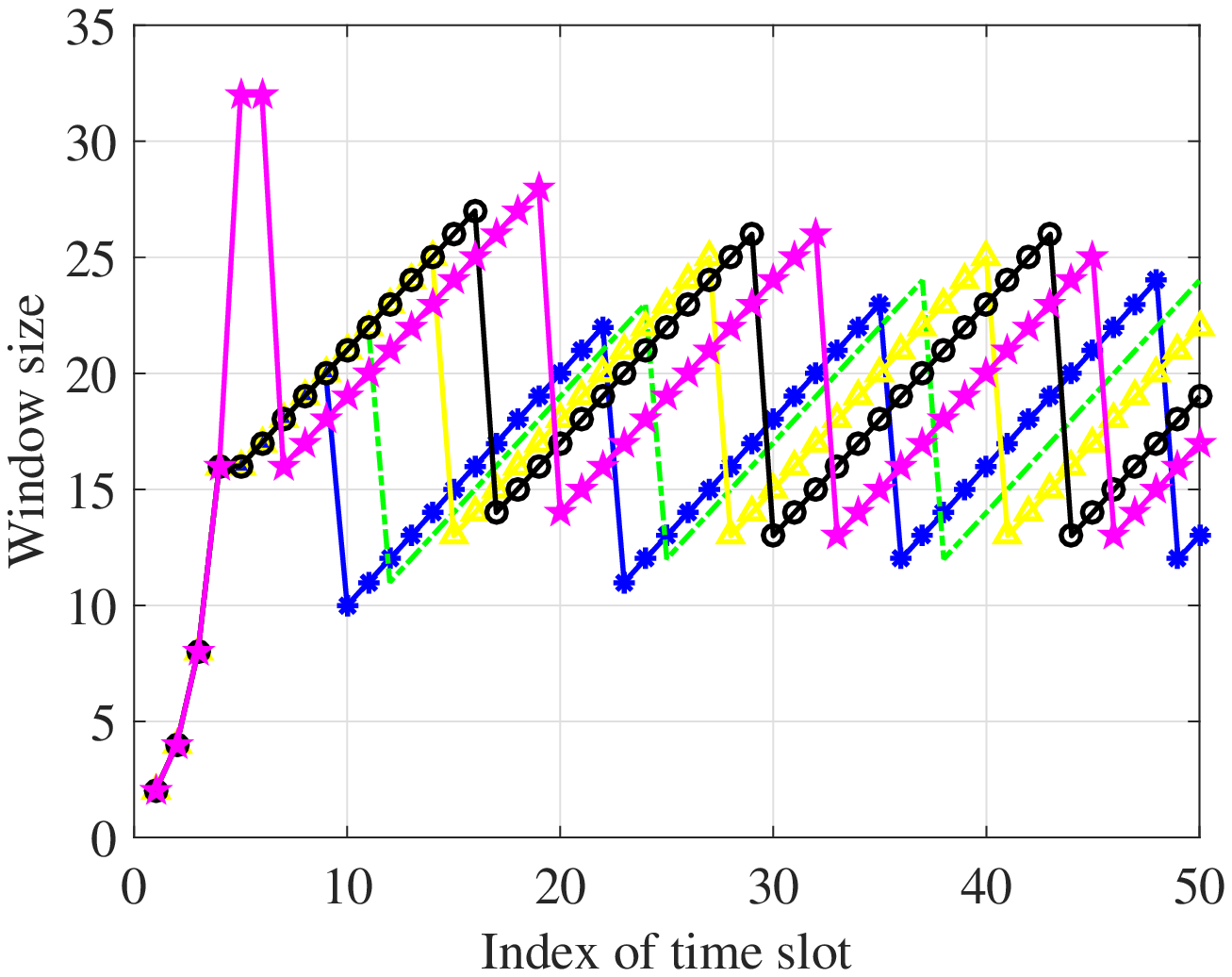}
	}
	\subfigure[How memory utilization changes.]{\label{subfig:microUtilization}
		\includegraphics[width=0.31\linewidth]{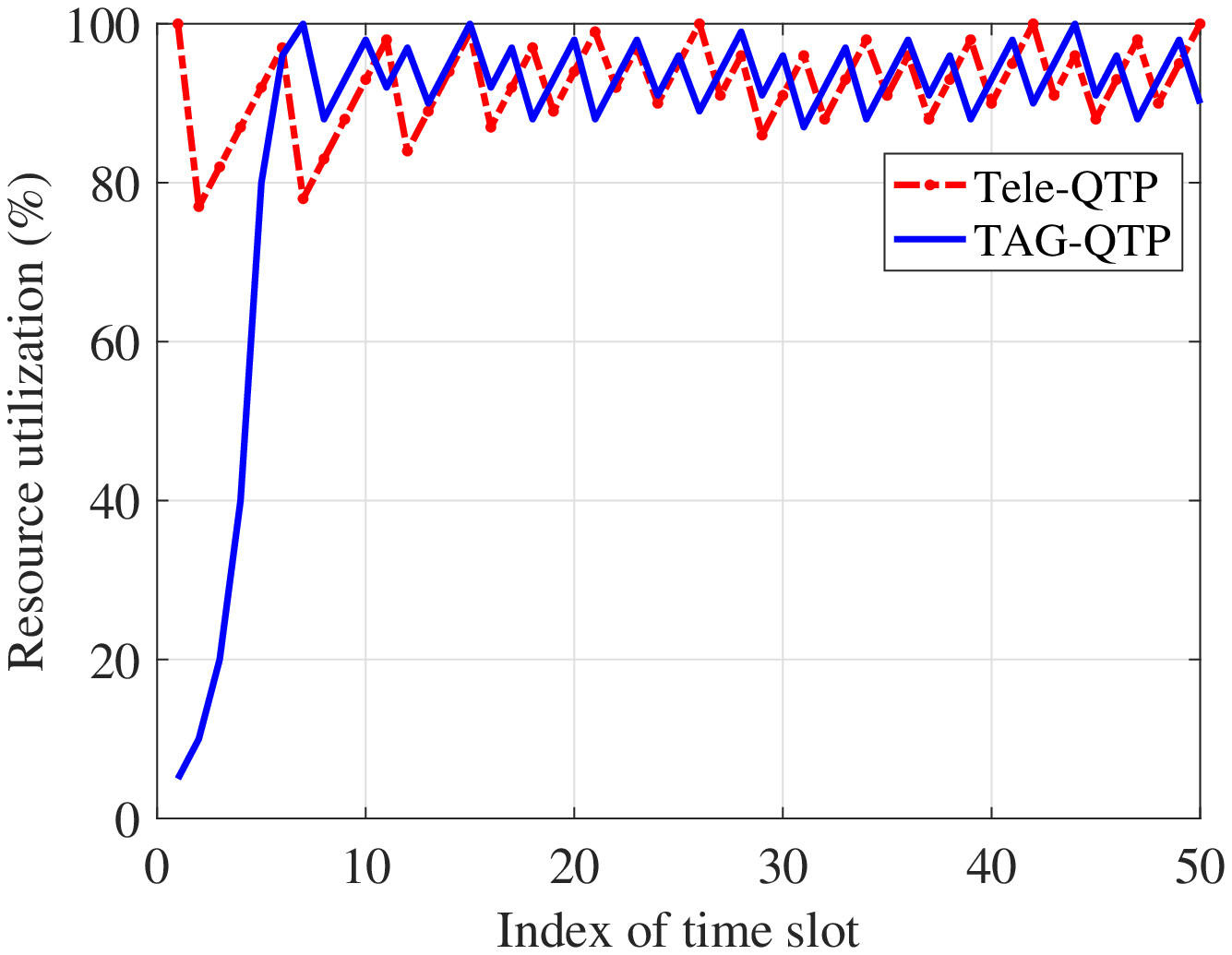}
	}
	\vspace{-0.1in}
	\caption{QTP performance when sharing a single bottleneck.}\label{fig:micro}
	\vspace{-0.2in}
\end{figure*}
To proof Theorem~\ref{the:fairness}, we leverage a fluid model (\ie, with continuous quantum resource amount and sending window sizes) as in~\cite{anaTCP} for simplicity. At first, we propose the following lemmas:

\begin{lemma}\label{lem:bound}
	In the steady state, suppose the sending window sizes of all the QTP sessions satisfy $W_1 \geq W_2 \geq \cdots \geq W_N$, we have
	$$\frac{W_1}{2} \leq W_N$$
\end{lemma}
\begin{proof}
	If $\frac{W_1}{2} > W_N$, $W_N$ will not decrease before a session's window is cut to be a size smaller than $W_N$. Therefore, there must be some time slot in which Lemma~\ref{lem:bound} holds. From this time slot on, consider a time slot that we need to cut the window of session $n$, its window size will become $\frac{W_n}{2}$, and session $n+1$ will have the largest window size. After $K$ time slots, we have to cut the window of session $n+1$. Since  $W_n \geq W_{n+1}$, there would be $\frac{W_n}{2}+K \geq \frac{W_{n+1}+K}{2}$ for any $K\geq 0$. That is the session with largest window size will become the session with smallest window size at the time slot when its window size is to be cut by half. This means Lemma~\ref{lem:bound} will always hold.
\end{proof}

\begin{lemma}\label{lem:winGap}
	In the steady state, suppose the sending window sizes of all the QTP sessions satisfy $W_1 \geq W_2 \geq \cdots \geq W_N$, with fluid model we know if $\frac{W_1}{2} = KN$, then
	$$W_n - W_{n+1} = K$$
	for all $n<N$.
\end{lemma}
\begin{proof}
	Consider the time at which we should reduce $W_1$ to avoid congestion, it will lease $\frac{W_1}{2}$ units of resources, which will be shared by $N$ QTP sessions. Accordingly, we will cut window again in $\frac{W_1}{2N} = K$ time slots. Now, the sending window of session 1 is $\frac{W_1}{2} + K$, while it is $\frac{W_2 + K}{2}$ for session 2. From Lemma~\ref{lem:bound}, we know $W_N \geq \frac{W_1}{2} + K \geq \frac{W_2 + K}{2}$. By investigating how the window size difference between session 1 and session 2 changes
	\begin{displaymath}
	\begin{split}
	\Delta & = (W_1 - W_2) - (\frac{W_1}{2} + K - \frac{W_2 + K}{2})\\
	& = \frac{W_1 - W_2}{2} + \frac{K}{2}
	\end{split}
	\end{displaymath}
	When $W_1 - W_2 > K$, $\Delta > 0$, \ie, the window size gap will decrease, while it will increase if $W_1-W_2 < K$ since $\Delta < 0$. Accordingly, in the steady state $W_1 - W_2 = K$. Above inference can be extended to $W_n - W_{n+1}$ for all $n<N$, which concludes the proof.
\end{proof}
Now, we are ready to prove Theorem~\ref{the:fairness}.
\begin{proof}
	From Lemma~\ref{lem:winGap}, let $W^*=W_1^{(t)} \geq W_2^{(t)} \geq \cdots \geq W_N^{(t)}$ be the sending window sizes at time slot $t$, $\frac{W^*}{2}=KN$, we know that at time slot $t+K$, there will be $W_{n+1}^{(t+K)} = W_{n+1}^{(t)} + K = W_n^{(t)}$. Accordingly, the window size combination is varying with a period of $K$ time slots, and the window size of each session is varying with a period of $KN$ time slots, \ie, window size of every QTP session is sawtoothly varying on $[W^*/2, W^*]$, and the average window size is $\frac{W^* + W^*/2}{2} = 3W^*/4$. 
\end{proof}

\section{Proof of Theorem~\ref{the:converving}}\label{app:conserving}
\begin{proof}
	From Theorem~\ref{the:fairness}, we know the average window size of each QTP session is $\frac{3W^*}{4}$. Since there is no congestion in the system (Theorem~\ref{the:congestion}), we have $\frac{3W^*}{4} \times N \leq C$. That is $W^* \leq \frac{4C}{3N}$. Then, in a specific time slot, there are at most $\frac{W^*}{2} = 2C/3N$ units, \ie, $2C/3CN \times 100 = 200/3N$ percent, of quantum memory sitting idle. 
\end{proof}

\section{Microcosmic Simulation to Validate Analysis in  Section~\ref{subsec:analysis}}\label{app:validation}

To show the correctness of analysis in Section~\ref{subsec:analysis}, we deploy a many-to-one network with 5 ingresses, each of which sets up a QTP session goes to the same egress, which becomes a bottleneck. In addition, we assume that the egress node has 100 units of quantum memory all of which are used for receiving data qubits. When studying \tele, we start the five QTP sessions at different sending window sizes which may occur when there are new QTP sessions arrive, while in a \shqdn, all the QTP sessions start at SS state and the initial window sizes are all 2, which emulates the case when multiple connections start simultaneously. The simulation results are shown in Fig.~\ref{fig:micro}. From this figure, we can make following observations.

First, from Fig.~\ref{subfig:microTele} \& \ref{subfig:microPlain}, we can see that with both QTPs, no matter what the initial sending window sizes are, they will converge to almost the same range. During the first 100 time slots, the Jain's fairness index on the average sending window size of all the five QTP sessions is 0.9914 if in the \tpqdn, while it is 0.9984 in \shqdn. Both of them are very close to the optimal fairness value 1. This coincides with the statement of Theorem~\ref{the:fairness}.

Second, from Fig.~\ref{subfig:microUtilization}, we can see that though the sending window of all the sessions in a \shqdn start from 2, the quantum memory utilization increases very quickly and achieves a high utilization level since sending window size of each session increases by 2x in every time slot. In both QDNs, when the quantum memory utilization approaches 1, the utilization will oscillate at a high level. After 10 time slots, the minimum quantum memory utilization at the egress node is about 87\%, which there is at most $\frac{200}{3 \times 5} \approx 13$ percent of the quantum memory sitting idle. This shows the correctness of Theorem~\ref{the:converving}.

In addition, during the entire simulation, we do not observe any congestion, which is the expectation according to Theorem~\ref{the:congestion}.


\end{document}